\documentclass[11pt]{amsart}
\usepackage{amsmath, amssymb}
\usepackage{amsfonts}
\usepackage{mathrsfs}
\usepackage[arrow,matrix,curve,cmtip,ps]{xy}
\usepackage{epsfig}
\usepackage{amsthm}
\usepackage{fullpage} 
\usepackage{hyperref}
\hypersetup{colorlinks=true,
linkcolor=blue, citecolor=blue, urlcolor=blue}
\usepackage{fullpage}

\usepackage[numbers]{natbib}

\usepackage{tikz-cd}

\allowdisplaybreaks

\newtheorem{theorem}{Theorem}[section]
\newtheorem{lemma}[theorem]{Lemma}
\newtheorem{proposition}[theorem]{Proposition}
\newtheorem{corollary}[theorem]{Corollary}

\newtheorem*{theorem*}{Theorem}
\theoremstyle{remark}
\newtheorem{remark}[theorem]{Remark}
\newtheorem{definition}[theorem]{Definition}
\newtheorem{example}[theorem]{Example}

\newcommand{\N}{\mathbb{N}}
\newcommand{\R}{\mathbb{R}}
\newcommand{\C}{\mathbb{C}}

\newcommand{\vp}{\varphi}

\newcommand{\norm}[1]{\left\lVert #1 \right\rVert}
\DeclareMathOperator{\Span}{span}

\title{A universal Representation for quantum commuting correlations}

\author{Roy Araiza}
\author{Travis Russell}
\author{Mark Tomforde}

\address{Department of Mathematics \& IQUIST, University of Illinois, Urbana-Champaign, IL, 61801 \\ USA}
\email{raraiza@illinois.edu}
\address{Department of Mathematics, Dartmouth College, Hanover, NH, 03755 \\ USA}
\email{travis.b.russell@dartmouth.edu}
\address{Department of Mathematics \\ University of Colorado \\ Colorado Springs, CO, 80918 \\USA}
\email{mtomford@uccs.edu}


\thanks{The first author was an Andrews Fellow supported by the Department of Mathematics, Purdue University.  The third author was supported by a grant from the Simons Foundation (\#527708 to Mark Tomforde)}

\begin{document}
\maketitle

\begin{abstract}
We explicitly construct an Archimedean order unit space whose state space is affinely isomorphic to the set of quantum commuting correlations. Our construction only requires fundamental techniques from the theory of order unit spaces and operator systems. Our main results are achieved by  characterizing when a finite set of positive contractions in an Archimedean order unit space can be realized as a set of projections on a Hilbert space.
\end{abstract}

\section{Introduction}

The past decade has witnessed a tremendous surge of interest in the theory of quantum correlations~---~probability distributions arising from independent measurements of entangled quantum systems. About ten years ago, several authors (e.g., \cite{junge2011connes}, \cite{fritz2012tsirelson}, and \cite{ozawa2013connes}) uncovered deep connections between the nearly fifty-year-old Connes' embedding problem \cite{connes1976classification} in operator algebras and the weak Tsirelson problem \cite{cirel1980quantum} concerning quantum correlations. In subsequent years, the literature on quantum correlations expanded rapidly as mathematicians, physicists, and computer scientists worked in tandem to address the Tsirelson's problems. This flurry of activity yielded several important results, including the non-closure of the set of finite-dimensional quantum correlations \cite{slofstra2019set} and the equality of the complexity classes $\operatorname{MIP}^*$ and $\operatorname{RE}$ established in \cite{ji2020mip}. The latter of these two results implies a negative solution to Connes' embedding problem.

Much of the literature on correlation sets focuses on distinguishing two classes of correlations: the class of finite-dimensional quantum correlations, denoted $C_q(n,k)$, and the class of quantum commuting correlations, denoted $C_{qc}(n,k)$, where $n$ and $k$ are parameters that denote the number of experiments and the number of outcomes, respectively, in a measurement scenario. While these sets are known to be convex and satisfy $C_q(n,k) \subseteq C_{qc}(n,k)$, a detailed description of their geometry has only been obtained in certain restrictive scenarios (e.g., \cite{Goh_et_al_2018}). Indeed, a recent preprint \cite{Fu_Miller_Slofstra_2021_preprint} shows that the problem of determining whether or not a given probability distribution belongs to the set of quantum correlations is undecidable. Consequently, new descriptions of the quantum correlation sets beyond their original definitions are valuable. One example of such a description is found in \cite{Navascu_s_2008}, which characterized the set $C_{qc}(n,k)$ as the set of probability distributions that can be certified by a certain infinite hierarchy of semidefinite programs. Another example can be found in \cite{lupini2020perfect}, where the set $C_{qc}(n,k)$ and the closure of the set $C_q(n,k)$ are separately identified with the state spaces of certain tensor products of finite-dimensional operator systems. These operator systems arise as subsystems of non-amenable universal group C*-algebras.

In this paper, we provide a new description for the set of quantum commuting correlations. Specifically, we explicitly construct a finite-dimensional ordered vector space whose state space is affinely isomorphic to the set $C_{qc}(n,k)$ of all quantum commuting correlations (see Theorem \ref{thm: qc correlation characterization}). This is achieved using techniques from the theory of operator systems. Moreover, our construction proceeds without reference to either universal C*-algebras or to the hierarchy of operator system tensor products from \cite{kavruk2011tensor}. Instead, we build upon previous work of the first two authors \cite{araiza2020abstract} which abstractly characterized projections as elements of operator systems. Our construction is achieved by first characterizing when a unital ordered vector space with a finite set of positive contractions $\{p_1, p_2, \dots, p_n\}$ can be realized as a subspace of the bounded operators on a Hilbert space $H$ such that each vector in $\{p_1, p_2, \dots, p_n\}$ is a projection on $H$ (see Theorem \ref{thm: Multi projections AOU characterization}). We then construct, for each $n$ and $k$, an Archimedean order unit space $(\mathcal{V}_{ns}, D_{ns},e_{ns})$ whose state space is affinely isomorphic to the set of nonsignalling correlations $C_{ns}(n,k)$ (see Theorem \ref{Characterize non-signalling correlations}). Using our results on projections in ordered vector spaces, we complete the positive cone $D_{ns}$ of $\mathcal{V}_{ns}$ to a positive cone $D_{qc}$, which arises from projection-valued measures on a Hilbert space, yielding the following theorem:

\begin{theorem} \label{thm: summary main theorem}
Let $n,k \in \mathbb{N}$ and let $\mathcal{V}_{ns}$ denote the corresponding universal nonsignalling vector space with generators $Q(a,b|x,y)$ where $a,b \in \{1,\dotsc,k\}$, and $x,y \in \{1,\dotsc,n\}$. If $p = \{p(a,b|x,y)\}$ is a correlation, then $p \in C_{qc}(n,k)$ (respectively, $p \in C_{ns}(n,k)$) if and only if there exists a state $\phi:(\mathcal{V}_{ns}, D_{qc}, e_{ns}) \to \mathbb{C}$ (respectively, a state $\phi:(\mathcal{V}_{ns}, D_{ns}, e_{ns}) \to \mathbb{C}$) such that $p(a,b|x,y) = \phi(Q(a,b|x,y))$ for each $a,b \in \{1,\dotsc,k\}$ and $x,y \in \{1,\dotsc,n\}$.
\end{theorem}

The main difficulty in producing these results lies in characterizing sets of projections as elements of Archimedean order unit spaces (Theorem \ref{thm: Multi projections AOU characterization}). The results of \cite{araiza2020abstract} characterize projections as elements of operator systems by finding conditions on an element $p$ of an operator system $\mathcal{V}$ which guarantee the existence of a unital complete order embedding $\pi: \mathcal{V} \to B(H)$ for which $\pi(p)$ is a projection on $H$. The authors then show that every element $p$ satisfying these conditions is a projection in the C*-envelope of $\mathcal{V}$. Since there is no analogue of the C*-envelope for an Archimedean order unit space, it becomes necessary to find conditions on a finite set of elements $\{p_1, p_2, \dots, p_n\}$ in an operator system $\mathcal{V}$ which guarantee the existence of a unital complete order embedding $\pi: \mathcal{V} \to B(H)$ for which $\pi(p_i)$ is a projection on $H$ for every $i \in \{1,2,\dots,n\}$ (see Theorem \ref{thm: Multi projection operator system characterization}). Moreover, to pass this result from the setting of operator systems to the setting of Archimedean order unit spaces requires the use of inductive limits of matrix orderings (see Definition \ref{defn: union of matrix cones}). This provides enough tools to prove Theorem \ref{thm: Multi projections AOU characterization}, from which our main results follow.

We conclude this introduction with a comparison of our results to those of \cite{lupini2020perfect}. The existence of an AOU space whose state space is affinely isomorphic to the set of quantum commuting correlations follows already from the early work of Kadison in \cite{kadison1951representation}, which detailed the duality between convex sets and AOU spaces. In \cite{lupini2020perfect}, this AOU space is identified with the first matrix level of the operator system $\mathcal{S}(n,k) \otimes_c \mathcal{S}(n,k)$, where $\mathcal{S}(n,k)$ is the operator system spanned by the generators of the $n$-fold free product $G_{n,k}=\mathbb{Z}_k * \dots * \mathbb{Z}_k$ in the corresponding universal free group C*-algebra $C^*(G_{n,k})$, and $\otimes_c$ is the commuting operator system tensor product developed in \cite{kavruk2011tensor}. Our work implies that this AOU space can also be identified with the space $(\mathcal{V}_{ns}, D_{qc}, e_{ns})$ described in Theorem \ref{thm: summary main theorem}. The principle difference between our result and the results of \cite{lupini2020perfect} is that we provide an explicit construction for the cone $D_{qc}$ from elementary order-theoretic ingredients, whereas the operator system $\mathcal{S}(n,k) \otimes_c \mathcal{S}(n,k)$ is constructed using elements of the non-ameanable universal group C*-algebra $C^*(G_{n,k})$ as well as the commuting tensor product, and taking advantage of their universal properties. Group C*-algebras for non-ameanable groups are not easily constructed or well-understood from elementary ingredients. The same can be said of the commuting tensor product of operator systems, even when the operator systems are finite dimensional. Our constructive approach provides a different perspective which may shed more light not only on the set of quantum commuting correlations, and also on the structure of the operator system $S_{n,k}$ and the commuting tensor product $\mathcal{S}(n,k) \otimes_c \mathcal{S}(n,k)$. For a more detailed discussion, we refer the reader to Section \ref{sec: qc correlations as states in AOU spaces}.

The organization of the paper is as follows: Section~\ref{sec: preliminaries} covers preliminary material from the theory of operator systems and Archimedean order unit spaces. Section~\ref{sec: single projections in AOU spaces} discusses Archimedean order unit spaces containing a single non-trivial projection. Section~\ref{sec: multiple projections in AOU spaces} establishes results concerning finite sets of projections in Archimedean order unit spaces. We conclude in Section~\ref{sec: qc correlations as states in AOU spaces} with applications of our results to the theory of quantum correlation sets.

\section{Preliminaries}\label{sec: preliminaries} 

Let $\mathbb{N}, \mathbb{R}$, and $\mathbb{C}$ denote the sets of natural, real and complex numbers, respectively. Given $n \in \mathbb{N}$, we let $[n] := \{1,2, \dots, n\}$. Throughout the paper we will be working in a variety of vector spaces, all of which may be assumed to be vector spaces over the field of complex numbers unless specified otherwise. Given a vector space $\mathcal{V}$, we let $M_n(\mathcal{V})$ denote the vector space of $n \times n$ matrices with entries in $\mathcal{V}$. In particular, $M_n(\mathbb{C})$ denotes the algebra of $n \times n$ matrices over $\mathbb{C}$. By a \textit{$*$-vector space}, we mean a vector space $\mathcal{V}$ equipped with a conjugate linear involution $*: \mathcal{V} \to \mathcal{V}$. The \textit{hermitian} elements of $\mathcal{V}$ are those elements $x$ satisfying $x = x^*$, and we let $\mathcal{V}_h$ denote the real vector space of hermitian elements. If $A = (a_{ij}) \in M_n(\mathcal{V})$, then $A^*$ will denote the conjugate transpose of $A$; i.e., the matrix whose $(i,j)$ entry is $a_{ji}^*$. For each $n \in \mathbb{N}$, we let $I_n \in M_n(\mathbb{C})$ denote the identity matrix and $J_n \in M_n(\mathbb{C})$ denote the matrix with every entry equal to 1. If $A \in M_n(\mathbb{C})$ and $B \in M_m(\mathbb{C})$, we let $A \oplus B \in M_{n+m}(\mathbb{C})$ denote the direct sum; i.e.,
\[ A \oplus B := \begin{pmatrix} A & 0 \\ 0 & B \end{pmatrix}. \]
\noindent We often use the canonical shuffle map $\vp: M_n(\mathbb{C}) \otimes M_m(\mathbb{C}) \to M_m(\mathbb{C}) \otimes M_n(\mathbb{C})$, which acts on elementary tensors via $\vp: A \otimes B \to B \otimes A$. This mapping is a $*$-isomorphism of the C*-algebra $M_{nm}(\mathbb{C})$ and extends to the corresponding shuffle maps on $M_n(\mathcal{V}) \otimes M_m(\mathbb{C})$ and $M_n(\mathbb{C}) \otimes M_m(\mathcal{V})$ via the identifications $M_n(\mathcal{V}) \otimes M_m(\mathbb{C}) \cong M_{nm}(\mathbb{C}) \otimes \mathcal{V} \cong M_n(\mathbb{C}) \otimes M_m(\mathcal{V})$.  

An \textit{ordered $*$-vector space} is a pair $(\mathcal V, C)$ where $\mathcal V$ is a $*$-vector space and $ C \subseteq \mathcal V_h$ is a positive cone; i.e., $C + C \subseteq C$ and $\R^+C \subseteq C$. A cone $C$ is \textit{proper} if $C \cap -C = \{0\}$, and when this occurs we call $(\mathcal V, C)$ a \textit{proper ordered $*$-vector space}. Given any ordered $*$-vector space $(\mathcal V, C)$ we may define $\mathcal{J} := \Span C \cap -C$ and form the proper ordered $*$-vector space $\mathcal{V} / \mathcal{J}$ with the proper ordering $C + \mathcal{J}$. For any proper ordered $*$-vector space $(\mathcal V, C)$, the positive cone $C$ induces a partial ordering on $\mathcal V_h$ by declaring $x \leq y$ if and only if $y - x \in C.$ An element $e \in \mathcal V_h$ is called an \textit{order unit} if for all $v \in \mathcal V_h$ there exists $t >0$ such that $te - v \in C.$ An \textit{Archimedean order unit} is an order unti $e$ that also satisfies the \textit{Archimedean property}: if $v \in \mathcal V$ and if $\epsilon e + v \in \mathcal C$ for all $\epsilon >0$, then $v \in \mathcal C$. A short argument shows that an Archimedean order unit must be an element of $C$. These properties imply that the positive cone majorizes the hermitian elements of $\mathcal V$ and that the positive cone $C$ is Archimedean closed. 
\begin{definition}
An \textit{Archimedean order unit (AOU)} space is a triple $(\mathcal V, C, e)$, where $(\mathcal V, C)$ is a proper ordered $*$-vector space and $e$ is an Archimedean order unit. 
\end{definition}  
\noindent If $e$ is an Archimedean order unit for $(\mathcal{V},C)$ but $C$ fails to be proper, we may form $\mathcal J:= \Span C \cap -C$ and then $(\mathcal{V} / \mathcal{J}, C + \mathcal{J}, e + \mathcal{J})$ is an AOU space (see Proposition~\ref{prop: nonproper operator system}). A linear map $\vp: (\mathcal V, C, e) \to (\mathcal W, D, f)$ between AOU spaces will be called \textit{positive} if $\vp(C) \subseteq D$, and \textit{unital} if $\vp(e) = f$. If $\vp$ is a (unital) linear isomorphism such that $\vp$ and $\vp^{-1}$ are both positive then we call $\vp$ a \textit{(unital) order isomorphism}. Given an AOU space $(\mathcal V, C, e)$ there is a canonical norm  $\| \cdot \|: \mathcal V_h \to [0,\infty)$ defined on the hermitian elements and given by $\norm{v} := \inf \{ t>0: te \pm v \in \mathcal C\}$.  We call this norm the \textit{order norm} associated with $e$. 

A \textit{function system} is defined to be a self-adjoint unital subspace of $C(K)$ for some compact Hausdorff space $K$. Due to the results stemming from \cite{kadison1951representation, paulsen2009vector}, every AOU space $\mathcal V$ may be identified as a concrete function subsystem of the continuous functions on the state space of $\mathcal V.$ In particular, every AOU space may be identified with the continuous affine functions on its state space, and conversely if $K$ is a compact convex subset of a locally convex space, and if $A(K)$ denotes the space of continuous affine functions on $K$, then $K$ is affinely isomorphic to the state space of $A(K)$.

Let $\mathcal V$ be a $*$-vector space. We define a \textit{matrix ordering} to be a sequence $\mathcal C:= \{\mathcal C_n\}_{n \in \N}$, where $\mathcal C_n$ is a cone in $M_n(\mathcal V)_h$ for each $n \in \N$ and, for every $\alpha \in M_{mn}(\mathbb{C})$, the inclusion $\alpha \mathcal{C}_n \alpha^* \subseteq \mathcal{C}_m$ is valid. We call $\mathcal C$ \textit{proper} if each  $\mathcal{C}_n$ is a proper cone. The pair $(\mathcal V, \mathcal C)$ is a \textit{(proper) matrix ordered $*$-vector space} if $\mathcal V$ is a $*$-vector space and $\mathcal C$ is a (proper) matrix ordering. The matrix ordering $\mathcal C$ defines a partial ordering on $M_n(\mathcal V)_h$ for each $n \in \N$ by declaring $x \leq y$ if and only if $y -x \in \mathcal C_n.$ An element $e \in \mathcal V_h$ will be called a \textit{matrix order unit} if $I_n \otimes e$ is an order unit for the ordered $*$-vector space $(M_n(\mathcal{V}), \mathcal{C}_n)$ for each $n \in \N$. It is a fact that $e \in \mathcal C_1$ is an order unit if and only if $e$ is a matrix order unit (see, e.g., \cite[Proposition 2.4]{araiza2020abstract}). If $I_n \otimes e$ is an Archimedean order unit for $(M_n(\mathcal{V}), \mathcal{C}_n)$ for each $n \in \mathbb{N}$, then we call $e$ an \textit{Archimedean matrix order unit}. This latter property ensures that each cone $\mathcal C_n$ is Archimedean closed.

\begin{definition}
An \textit{operator system} is a triple $(\mathcal V, \mathcal C, e)$, where $\mathcal V$ is a $*$-vector space, $\mathcal C$ is a proper matrix ordering, and $e$ is an Archimedean matrix order unit. 
\end{definition}
\noindent When no confusion arises we will simply denote an operator system by $\mathcal V$. If $u: \mathcal V \to \mathcal W$ is a linear map between operator systems, we define the \textit{$n$\textsuperscript{th}-amplification} of $u$ to be the map $u_n:= I_n \otimes u: M_n(\mathcal V) \to M_n(\mathcal W)$ defined by $\sum_{ij} e_ie_j^* \otimes v_{ij} \mapsto e_ie_j^* \otimes u(v_{ij})$, where $\{e_i\}_{i=1}^n \subseteq \C^n$ denotes the canonical column basis vectors. Letting $\mathcal C$ and $\mathcal D$ be the respective proper matrix orderings on $\mathcal V$ and $\mathcal W$, we call the map $u$ \textit{completely positive} if $u_n(\mathcal C_n) \subseteq \mathcal D_n$ for each $n \in \N.$ If $u: \mathcal V \to \mathcal W$ is a  (unital) linear isomorphism such that both $u$ and $u^{-1}$ are completely positive then we say that $u$ is a (unital) complete order isomorphism. When $u$ is a (unital) complete order isomorphism onto its range, we will sometimes call $u$ a \textit{(unital) complete order embedding}. We will identify two operator systems if there exists a unital complete order isomorphism between them.  A \textit{concrete operator system} is defined to be a unital self-adjoint subspace of $B(H)$. A fundamental result from \cite[Theorem 4.4]{choi1977injectivity} states that for any operator system $(\mathcal V, \mathcal C, e)$ then there exists a Hilbert space $H$ and a concrete operator system $\widetilde{\mathcal V} \subseteq B(H)$ such that $\mathcal V$ is unital completely order isomorphic to $\widetilde{\mathcal V}.$ Similar to the case for AOU spaces, it follows that operator systems may be identified with what are known as continuous matrix affine functions acting on the matrix state space of that operator system and the converse noncommutative analogue also holds. For the converse, one may identify objects known as compact matrix convex sets with the matrix state space of the continuous matrix affine functions acting on that matrix convex set. This duality is known as Webster-Winkler duality and was first investigated in \cite{webster1999krein}.

Throughout the manuscript we will be dealing with matrix orderings that are not a priori proper.  For such a matrix ordering, we may always consider a natural quotient of the matrix ordered space such that the quotient is necessarily a proper matrix ordered $*$-vector space. In particular, given any matrix ordered $*$-vector space $(\mathcal V, \mathcal C)$ we consider the quotient $\mathcal V/ \mathcal J,$ where $\mathcal J:= \Span \mathcal C_1 \cap -\mathcal C_1.$ With the natural involution defined on cosets as $(v + \mathcal J)^* := v^* + \mathcal J$ and letting $\mathcal C + \mathcal J:= \{\mathcal C_n + M_n(\mathcal J)\}_{n \in \N}$, it necessarily follows that $(\mathcal V/ \mathcal J, \mathcal C + \mathcal J)$ is a proper matrix ordered $*$-vector space. Given any matrix ordered $*$-vector space $(\mathcal V, \mathcal C)$ with Archimedean matrix order unit $e$, it also follows that $e + \mathcal J$ is an Archimedean matrix order unit for the proper matrix ordered $*$-vector space $(\mathcal V/ \mathcal J, \mathcal C + \mathcal J).$ In particular: 
\begin{proposition}[{\cite[Proposition 4.4]{araiza2020abstract}}] \label{prop: nonproper operator system}
Given any matrix ordered $*$-vector space $(\mathcal V, \mathcal C)$ with Archimedean matrix order unit $e$, the triple $(\mathcal V/ \mathcal J, \mathcal C + \mathcal J, e + \mathcal J)$ is an operator system.
\end{proposition}

Given an operator system $\mathcal V$, it was shown in \cite{hamana1979injective} that there exists a C*-algebra $\mathcal A$ and a unital complete order embedding $j: \mathcal V \to \mathcal A$ such that $\mathcal A= C^*(j(\mathcal V))$ and $\mathcal A$ satisfies the following universal property: given a pair $(\mathcal D, i)$ where $i: \mathcal V \to \mathcal D$ is a unital complete order embedding and $\mathcal D$ is a (unital) C*-algebra generated by $i(\mathcal V)$, then there exists a unique surjective $*$-homomorphism $\sigma: \mathcal D \to \mathcal A$ such that $\sigma \circ i = j.$ Thus we have the following commutative diagram: 

\[
\begin{tikzcd}
\mathcal D \arrow[dashrightarrow, rd, "\sigma" black] \\
\mathcal V \arrow[u, "i"] \arrow[r, "j" black] & \mathcal A
\end{tikzcd}
\]

\noindent The C*-algebra $\mathcal A$ is called the \textit{C*-envelope} of $\mathcal V$ and we will often denote it by $C_e^*(\mathcal V).$ Given an operator system $\mathcal V$, a pair $(\mathcal D, i)$ consisting of a C*-algebra $\mathcal D$ and a unital complete order embedding $i: \mathcal V \to \mathcal D$ is called a \textit{C*-cover} of $\mathcal V$ when $\mathcal D = C^*(i(\mathcal V)).$ Thus the C*-envelope of an operator system may be viewed as the minimal C*-cover of that operator system.

In \cite{paulsen2011operator}, it was shown that given any AOU space $\mathcal V$ there exists a maximal operator system structure on $\mathcal V$. Given an AOU space $(\mathcal V, C, e)$, for each $n \in \mathbb{N}$ we define
\[
D_n^{max} := \{ a^*va: a \in M_{m,n} \text{ and } v = \oplus_{i=1}^m v_i \text{ for some } m \in \mathbb{N} \text{ and for some } v_1, \ldots, v_m \in C \}.
\]
Then $D^{max} := \{ D_n^{max} \}_{n=1}^\infty$ is a proper matrix ordering on $\mathcal V$ with the property that if $\mathcal P$ is any other matrix ordering on $\mathcal V$ with $\mathcal P_1 = C$, then $D_n^{max} \subseteq \mathcal P_n$ for each $n \in \N.$  In general, the matrix order unit $e$ may fail to be Archimedean for $D^{max}$. Thus, for each $n \in \mathbb{N}$ we define
\[
C_n^{max}:=\{ v \in M_n(\mathcal V)_h: \forall \epsilon >0,\,\, v+ \epsilon(I_n \otimes e) \in D_n^{max}\}
\]
and form the proper matrix ordering $C^{max} := \{ C_n^{max} \}_{n=1}^\infty$. (The process of going from $D_n^{max}$ to $C_n^{max}$ is called \textit{Archimedeanization}; see \cite{paulsen2009vector}.)  The triple $(\mathcal V, C^{max},e)$ is an operator system called the \textit{maximal operator system structure} on $\mathcal V$.  For ease of notation, we shall denote the triple as $\mathcal V_{max}$ when no confusion will arise.

\section{Abstract projections in operator systems and AOU spaces}\label{sec: single projections in AOU spaces}

Given an operator system or an AOU space $\mathcal{V}$ with order unit $e$, and a positive contraction $p \in \mathcal{V}$, we let $p^\perp$ denote the positive contraction $e-p$.

\begin{definition}\label{Defn: abstract compression operator system from one contraction}
Suppose that $(\mathcal{V}, \mathcal{C}, e)$ is an operator system. Given a positive contraction $p \in \mathcal{V}$, we let $\mathcal{C}(p)$ denote the (generally non-proper) matrix ordering on $M_2(\mathcal{V})$ defined as follows: \[
\mathcal C(p)_n:= \{ x \in M_{2n}(\mathcal V)_h: \forall \epsilon>0 \,\,\exists t>0 \,\,\text{such that} \,\, x + \epsilon I_n \otimes (p \oplus p^\perp) + tI_n \otimes (p^\perp \oplus p) \in \mathcal C_{2n}\}.
\]
Let $\mathcal{J}_p := \Span C(p)_1 \cap -C(p)_1$, and define $\pi_p: \mathcal V \to M_2(\mathcal V)/\mathcal J_p$ by \[\pi_p(x)= x \otimes J_2 + \mathcal{J}_p. 
\] We consider $M_2(\mathcal V)/ \mathcal J_p$  with operator system structure $
(M_2(\mathcal{V})/ \mathcal J_p,\mathcal{C}(p) + \mathcal{J}_p, I_2 \otimes e + \mathcal{J}_p).$

\end{definition}

We remark that for each $n \in \mathbb{N}$ we identify  $M_n( M_2(\mathcal{V}) / \mathcal{J}_p)$ with $M_{2n}(\mathcal{V}) / M_n(\mathcal{J}_p)$. Moreover, for each $x \in M_n(\mathcal{V})$, the $n$\textsuperscript{th}-amplification $(\pi_p)_n: M_n(\mathcal{V}) \to M_{2n}(\mathcal{V}) / M_n(\mathcal{J}_p)$ satisfies $(\pi_p)_n(x) = x \otimes J_2 + M_n(\mathcal J_p)$.

\begin{proposition} \label{prop: pi_p is always ucp}
Let $(\mathcal{V}, \mathcal{C}, e)$ be an operator system, and suppose that $p \in \mathcal{V}$ is a positive contraction. Then the map $\pi_p: \mathcal{V} \to M_2(\mathcal{V}) / \mathcal{J}_p$ from Definition \ref{Defn: abstract compression operator system from one contraction} is a unital completely positive map.
\end{proposition}

\begin{proof}
Unitality follows from \cite[Lemma 5.6]{araiza2020abstract}. It remains to show complete positivity. Let $x \in \mathcal{C}_n$. Then $x \otimes J_2 \in \mathcal{C}_{2n}$. It follows that for every $\epsilon > 0$ we have
\[ x \otimes J_2 + \epsilon I_n \otimes (p \oplus p^{\perp}) + \epsilon I_n \otimes (p^{\perp} \oplus p) \in \mathcal C_{2n}. \]
Therefore $x \otimes J_2 \in \mathcal{C}(p)_n$. It follows that
\[ (\pi_p)_n(x) = x \otimes J_2 + M_n(\mathcal{J}_p) \in \mathcal{C}(p)_n + M_n(\mathcal{J}_p). \]
So $\pi_p$ is completely positive.
\end{proof}

While Proposition~\ref{prop: pi_p is always ucp} implies that $\pi_p$ is always completely positive, it is not necessarily an injective map in general, much less a complete order embedding. The importance of $\pi_p$ is illustrated by the following theorem.

\begin{theorem}[{\cite[Definition 5.4, Theorem 5.7, Theorem 5.8]{araiza2020abstract}}] \label{thm: main thm AR1.0}
Let $(\mathcal{V}, \mathcal{C}, e)$ be an operator system and suppose that $p \in \mathcal{V}$ is a positive contraction. Then the following statements are equivalent:
\begin{enumerate}
    \item The map $\pi_p: \mathcal V \to M_2(\mathcal V)/ \mathcal{J}_p$ is a complete order embedding. 
    \item There exists a Hilbert space $H$ and a unital complete order embedding $\pi: \mathcal{V} \to B(H)$ such that $\pi(p)$ is a projection in $B(H)$.
    \item The element $p$ is a projection in $C_e^*(\mathcal{V})$.
\end{enumerate}
\end{theorem}

\begin{definition}\label{Defn: abstract projection}
We say that a positive contraction $p$ in an operator system $(\mathcal{V}, \mathcal{C}, e)$ is an \textit{abstract projection} if $p$ satisfies one (and hence all) of the conditions in Theorem~\ref{thm: main thm AR1.0}. 
\end{definition}

\begin{remark}
Our notion of an abstract projection above differs slightly from the one in \cite{araiza2020abstract} in that we require that $p$ be a positive contraction, but do not require $\|p\| \in \{0,1\}$. However the requirements that $p$ is a positive contraction and that $\pi_p$ is a complete order embedding are sufficient. Indeed, it can be shown that if $0 < \|p\| < 1$ then $\pi_p(p) = 0$ and therefore $\pi_p$ is not an order embedding. 
\end{remark}

Notice that Condition~(3) of Theorem~\ref{thm: main thm AR1.0} implies that whenever $p$ and $q$ are abstract projections in an operator system $\mathcal{V}$, there always exists a Hilbert space $H$ and a unital complete order embedding $\pi: \mathcal{V} \to B(H)$ such that both $\pi(p)$ and $\pi(q)$ are projections on $H$. Indeed, this can be achieved by applying the Gelfand-Naimark Theorem to the C*-algebra $C_e^*(\mathcal{V})$.

We now turn our attention to projections in AOU spaces. We would like to characterize when a positive contraction $p$ in an AOU space $(\mathcal{V}, C, e)$ is an ``abstract projection'' in the sense that there exists a Hilbert space $H$ and a unital order embedding $\pi: \mathcal{V} \to B(H)$ such that $\pi(p)$ is a projection in $B(H)$. By Theorem~\ref{thm: main thm AR1.0} it suffices to characterize when there exists a matrix ordering $\mathcal{C}$ on $\mathcal{V}$ such that $(\mathcal{V}, \mathcal{C}, e)$ is an operator system, $\mathcal{C}_1 = C$, and $p$ is an abstract projection in $(\mathcal{V}, \mathcal{C}, e)$.

\begin{lemma} \label{lem: compression cones equal AOU case}
Let $(\mathcal{V}, C, e)$ be an AOU space. Suppose that $\mathcal{C}$ is a matrix ordering on $\mathcal{V}$ such that $\mathcal{C}_1 = C$ and $(\mathcal{V}, \mathcal{C}, e)$ is an operator system, and suppose that $p \in \mathcal{V}$ is a positive contraction. Define a matrix ordering $\mathcal{T}$ on $\mathcal{V}$ by 
\[ \mathcal{T}_n := \{ x \in M_n(\mathcal{V}) : x \otimes J_2 \in \mathcal{C}(p)_n \}. \] Then $\mathcal{T}(p) = \mathcal{C}(p)$.
\end{lemma}

\begin{proof}
First, we claim that $\mathcal{T}_n = (\pi_p)_n^{-1}(\mathcal{C}(p)_n + M_n(\mathcal{J}_p))_h$, and hence $\mathcal T$ is a matrix ordering on $\mathcal V$. The inclusion $\mathcal{T}_n \subseteq (\pi_p)_n^{-1}(\mathcal{C}(p)_n + M_n(\mathcal{J}_p))$ is clear by the definition of $\pi_p$. On the other hand, suppose that $x \in (\pi_p)_n^{-1}(\mathcal{C}(p)_n + M_n(\mathcal{J}_p))_h$. Then there exists $y \in M_n(\mathcal{J}_p)$ such that $x \otimes J_2 + y \in \mathcal{C}(p)_n$. Since $x=x^*$ and the elements of $\mathcal{C}(p)_n$ are self-adjoint, $y=y^*$. By \cite[Lemma 4.1 and Lemma 4.2]{araiza2020abstract}, we have $M_n(\mathcal{J}_p)_h = \mathcal{C}(p)_n \cap -\mathcal{C}(p)_n$. Therefore $(x \otimes J_2 + y) + (-y) \in \mathcal{C}(p)_n$. So $(\pi_p)_n^{-1}(\mathcal{C}(p)_n + M_n(\mathcal{J}_p))_h \subseteq \mathcal{T}_n$ and the claim justified.

By Proposition~\ref{prop: pi_p is always ucp}, we know that $\pi_p:\mathcal{V} \to M_2(\mathcal V)/ \mathcal{J}_p$ is completely positive. Therefore whenever $x \in \mathcal{C}_n$, we have $(\pi_p)_n(x) \in \mathcal{C}(p)_n + M_n(\mathcal{J}_p)$. It follows that $\mathcal{C}_n \subseteq \mathcal T_n$. Now suppose that $x \in \mathcal{C}(p)_n$. Then for every $\epsilon > 0$ there exists a $t > 0$ such that \[ x + \epsilon I_n \otimes  (p \oplus p^{\perp}) + tI_n \otimes (p^{\perp} \oplus p) \in \mathcal{C}_{2n}. \] But since $\mathcal{C}_{2n} \subseteq \mathcal T_{2n}$ for all $n$ we conclude that $x \in \mathcal{T}(p)_n$. Hence $\mathcal{C}(p)_n \subseteq \mathcal{T}(p)_n$. 

On the other hand, suppose that $x \in \mathcal{T}(p)_n$. Then for every $\epsilon > 0$ there exists $t > 0$ such that
\[ (\pi_p)_{2n}(x + \epsilon I_n \otimes (p \oplus p^{\perp}) + t I_n \otimes (p^{\perp} \oplus p)) \in \mathcal{C}(p)_{2n} + M_{2n}(\mathcal{J}_p). \] 
Hence, for every $\epsilon > 0$, there exist $t, r > 0$ such that
\[ \left[ x + \frac{\epsilon}{2} I_n \otimes (p \oplus p^{\perp}) + t I_n \otimes (p^{\perp} \oplus p) \right] \otimes J_2 + \frac{\epsilon}{2} I_{2n} \otimes (p \oplus p^{\perp}) + r I_{2n} \otimes  (p^{\perp} \oplus p) \in \mathcal{C}_{4n}. \]
Applying the canonical shuffle $\vp: M_{2n} \otimes M_2 \to M_2 \otimes M_{2n}$ to the expression above we obtain
\[ J_2 \otimes \left[ x + \frac{\epsilon}{2} I_n \otimes (p \oplus p^{\perp}) + t I_n \otimes (p^{\perp} \oplus p) \right] + \frac{\epsilon}{2} (p \oplus p^{\perp}) \otimes I_{2n}  + r (p^{\perp} \oplus p) \otimes I_{2n} \]
or, in matrix form,
\[ \begin{pmatrix} x & x \\ x & x \end{pmatrix} + \frac{\epsilon}{2} \begin{pmatrix} I_n \otimes (p \oplus p^{\perp}) & I_n \otimes (p \oplus p^{\perp}) \\ I_n \otimes (p \oplus p^{\perp}) & I_n \otimes(p \oplus p^{\perp}) \end{pmatrix} + t \begin{pmatrix} I_n \otimes (p^{\perp} \oplus p) & I_n \otimes (p^{\perp} \oplus p) \\ I_n \otimes (p^{\perp} \oplus p) & I_n \otimes (p^{\perp} \oplus p) \end{pmatrix} \] \[ \quad + \frac{\epsilon}{2} \begin{pmatrix} I_{2n} \otimes p & 0 \\ 0 & I_{2n} \otimes p^{\perp} \end{pmatrix} + r \begin{pmatrix} I_{2n} \otimes p^{\perp} & 0 \\ 0 & I_{2n} \otimes p \end{pmatrix}. \]
Let $W$ be the $4n \times 4n$ scalar permutation matrix which exchanges the $j$\textsuperscript{th} and $(j+2n)$\textsuperscript{th} columns for $j=2,4, \dots, 2n$. Then conjugating the final expression above by $W$ yields
\[ \begin{pmatrix} x & x \\ x & x \end{pmatrix} + \frac{\epsilon}{2} \begin{pmatrix} I_n \otimes (p \oplus p^{\perp}) & I_n \otimes (p \oplus p^{\perp}) \\ I_n \otimes (p \oplus p^{\perp}) & I_n \otimes(p \oplus p^{\perp}) \end{pmatrix} + t \begin{pmatrix} I_n \otimes (p^{\perp} \oplus p) & I_n \otimes (p^{\perp} \oplus p) \\ I_n \otimes (p^{\perp} \oplus p) & I_n \otimes (p^{\perp} \oplus p) \end{pmatrix} \] \[ \quad + \frac{\epsilon}{2} \begin{pmatrix} I_{n} \otimes (p \oplus p^{\perp}) & 0 \\ 0 & I_{n} \otimes (p^{\perp} \oplus p) \end{pmatrix} + r \begin{pmatrix} I_{n} \otimes (p^{\perp} \oplus p) & 0 \\ 0 & I_{n}   \otimes (p \oplus p^{\perp}) \end{pmatrix}. \]
Compressing to the upper left corner of this expression and using the compatibility of $\mathcal{C}$, we see that for every $\epsilon > 0$ there exists $t,r > 0$ such that 
\[ x + \epsilon I_n \otimes (p \oplus p^{\perp}) + (t+r)I_n \otimes (p^{\perp} \oplus p) \in \mathcal{C}_{2n}. \]
Thus $x \in \mathcal{C}(p)_n$. Therefore $\mathcal{T}(p)_n \subseteq \mathcal{C}(p)_n$. We conclude that for every $n \in \N$, $\mathcal{T}(p)_n = \mathcal{C}(p)_n$.
\end{proof}

\begin{theorem} \label{thm: characterize projections in AOU space}
Let $(\mathcal{V},C,e)$ be an AOU space, and suppose that $p \in \mathcal{V}$ is a positive contraction. Then the following statements are equivalent:
\begin{enumerate}
    \item There exists a matrix ordering $\mathcal{C}$ with $\mathcal{C}_1=C$ such that $p$ is an abstract projection in the operator system $(\mathcal{V},\mathcal{C},e)$.
    \item The map $\pi_p: (\mathcal{V}, C) \to (M_2(\mathcal{V}) / \mathcal{J}_{p}, C^{max}(p)_1 + \mathcal{J}_p)$ is an order embedding.
\end{enumerate}
\end{theorem}

\begin{proof}

Suppose that $(1)$ holds and that $\pi_p(x) \in C^{max}(p)_1 + \mathcal{J}_p$. Then for every $\epsilon > 0$ there exists a $t > 0$ such that \[ \begin{pmatrix} x & x \\ x & x \end{pmatrix} + \epsilon (p \oplus p^{\perp}) + t (p^{\perp} \oplus p) \in C_2^{max} \subseteq \mathcal{C}_2. \] It follows that $\pi_p(x) \in \mathcal{C}(p)_1 + \mathcal{J}_p$, and hence $x \in \mathcal{C}_1 = C$. On the other hand, if $x \in C$ then $\pi_p(x) \in C^{max}(p)_1 + \mathcal{J}_p$ by Proposition \ref{prop: pi_p is always ucp}. Hence (2) holds.

Next suppose that $(2)$ holds. Define a matrix ordering $\mathcal{C}$ on $\mathcal V$ by 
\[ \mathcal{C}_n := (\pi_p)_n^{-1}(C^{max}(p)_n + M_n(\mathcal{J}_p)). \]
Since $\pi_p$ is an order embedding, $\mathcal{C}_1 = C$. It remains to show that $\pi_p$ is a complete order embedding on $(\mathcal{V}, \mathcal{C})$. By Lemma \ref{lem: compression cones equal AOU case} we see that $\mathcal{C}(p)_n = C^{max}(p)_n$ for each $n$. Thus, if $(\pi_p)_{n}(x) \in \mathcal{C}(p)_n + M_n(\mathcal{J}_p)$, then $x \in \mathcal{C}_n$, since $\mathcal{C}(p)_n + M_n(\mathcal{J}_p) = C^{max}(p)_n+ M_n(\mathcal{J}_p)$ for each $n$. Therefore $\pi_p$ is a complete order embedding.
\end{proof}

Theorem~\ref{thm: characterize projections in AOU space} has an important limitation distinguishing it from Theorem~\ref{thm: main thm AR1.0}. Suppose that $p$ and $q$ are elements of an AOU space $(\mathcal{V}, C, e)$ both of which satisfy Condition~(2) of Theorem~\ref{thm: characterize projections in AOU space}. Then there exist matrix orderings $\mathcal{C}$ and $\mathcal{T}$ such that $(\mathcal{V}, \mathcal{C}, e)$ and $(\mathcal{V}, \mathcal{T}, e)$ are operator systems with $\mathcal{C}_1 = \mathcal{T}_1 = C$, $p$ is an abstract projection in $(\mathcal{V}, \mathcal{C}, e)$, and $q$ is an abstract projection in $(\mathcal{V}, \mathcal{T}, e)$. However it is not necessary that $\mathcal{C} = \mathcal{T}$ in general, nor is it clear how one could construct a new matrix ordering $\mathcal{R}$ with $\mathcal{R}_1 = C$ for which $p$ and $q$ are both abstract projections in $(\mathcal{V}, \mathcal{R}, e)$. Addressing this difficulty will be the subject of the next section.

\section{Multiple projections in AOU spaces}\label{sec: multiple projections in AOU spaces}

In Section~\ref{sec: single projections in AOU spaces} we described when a single contraction in an AOU space is a projection.  In this section we shall characterize when the elements of a finite set of contractions in an AOU space are projections.  To do this we first consider finite sets of abstract projections in an operator system.

\begin{definition} \label{defn: multiple projection compression cones}
Let $(\mathcal{V},\mathcal{C},e)$ be an operator system, and assume $p_1, p_2, \dots, p_N \in \mathcal{V}$ are positive contractions. For each $k=1,2,\dots,N$ define matrices $P_k^N, Q_k^N \in M_{2^N}(\mathcal{V})$ by
\[ P_k^N := I_{2^{k-1}} \otimes (p_k \oplus p_k^{\perp}) \otimes J_{2^{N-k}} \]
and
\[ Q_k^N := I_{2^{k-1}} \otimes (p_k^{\perp} \oplus p_k) \otimes J_{2^{N-k}} \]
Similarly, we define matrices $\widehat{P}_k^N, \widehat{Q}_k^N \in M_{2^N}(\mathcal{V})$ by
\[ \widehat{P}_k^N := J_{2^{k-1}} \otimes (p_k \oplus p_k^{\perp}) \otimes J_{2^{N-k}} \]
and
\[ \widehat{Q}_k^N := J_{2^{k-1}} \otimes (p_k^{\perp} \oplus p_k) \otimes J_{2^{N-k}}. \]
For each $n \in \mathbb{N}$, we define
\[ \mathcal{C}(p_1, p_2, \dots, p_N)_n \]
to be the set of $x \in M_{n2^N}(\mathcal{V})$ such that $x=x^*$ and for every $\epsilon_1, \epsilon_2, \dots, \epsilon_N > 0$ there exist $t_1, t_2, \dots, t_N > 0$ such that
\begin{equation} x + \sum_{k=1}^N \epsilon_k I_n \otimes P_k^N + \sum_{k=1}^N t_k I_n \otimes Q_k^N \in \mathcal{C}_{n2^N}, \label{eq: cones dependent on multiple contractions}
\end{equation} with the property that if one replaces $\epsilon_i$ with  $\epsilon_i^\prime < \epsilon_i$ then there exists $t_i^\prime > t_i$ such that Equation~\eqref{eq: cones dependent on multiple contractions} still holds. 
We define $\widehat{\mathcal{C}}(p_1, p_2, \dots, p_N)_n$ similarly with $\widehat{P}_k^N$ and $\widehat{Q}_k^N$ playing the role of $P_k^N$ and $Q_k^N$. Note that for each $n \in \N$ the sets $\mathcal C(p_1,\dotsc,p_N)_n$ and $\widehat{C}(p_1,\dotsc,p_N)_n$ are nonempty since if $x \in M_{n2^N}(\mathcal V)^+$ one may choose $t_i = \epsilon_i$ for $i=1,\dotsc,N$. Let $\mathcal{C}(p_1, p_2, \dots, p_N) := \{\mathcal{C}(p_1, p_2, \dots, p_N)_n\}_{n=1}^\infty$, and $\widehat{\mathcal{C}}(p_1, p_2, \dots, p_N):=\{\widehat{\mathcal{C}}(p_1, p_2, \dots, p_N)_n\}_{n=1}^\infty $.
\end{definition}

In the following we make use of the sequences of cones $\mathcal{C}(p_1, \dots, p_N)$ and $\widehat{\mathcal{C}}(p_1, \dots, p_N)$. In the case $N=1$ these sequences coincide and are both equal to the non-proper matrix ordering $\mathcal C(p)$, which was used to characterize single projections in \cite{araiza2020abstract} (cf. Section~\ref{sec: single projections in AOU spaces}). The cones $\widehat{\mathcal{C}}(p_1, \dots, p_N)$ are useful because they are symmetric with respect to the order of the operators $p_1, \dots, p_N$ in the sense that a rearrangement of these operators can be realized by simply applying a canonical shuffle to the corresponding operators $\widehat{P}_i^N$. However we will focus primarily on the cones $\mathcal{C}(p_1, p_2, \dots, p_N)$.

\begin{lemma} \label{lem: p-hats contained in p's}
If $N,n \in \N$, then $\widehat{\mathcal{C}}(p_1, p_2, \dots, p_N)_n \subseteq \mathcal{C}(p_1, p_2, \dots, p_N)_n$. 
\end{lemma}

\begin{proof}
First observe that for each $k \in \N$ we have
\begin{eqnarray}
 I_n \otimes \hat{P}_k^N & = & I_n \otimes J_{2^{k-1}} \otimes (p_k \oplus p_k^{\perp}) \otimes J_{2^{N-k}} \nonumber \\
& \leq & 2^{k-1} I_n \otimes I_{2^{k-1}} \otimes (p_k \oplus p_k^{\perp}) \otimes J_{2^{N-k}} \nonumber \\
& = & 2^{k-1} I_n \otimes P_k^N \nonumber
\end{eqnarray}
and similarly $I_n \otimes \hat{Q}_k^N \leq 2^{k-1} I_n \otimes Q_k^N$. Indeed, these inequalities hold because $J_2 \leq 2I_2$.

Next suppose $x \in \widehat{\mathcal{C}}(p_1, p_2, \dots, p_N)_n$. Then for every $\epsilon_1, \dots, \epsilon_N > 0$ there exist $t_1, \dots, t_N > 0$ such that
\[ x + \sum_{k=1}^{N} \frac{\epsilon_k}{2^{k-1}} I_n \otimes \widehat{P}_k^N + \sum_{k=1}^{N} t_k I_n \otimes \widehat{Q}_k^N \geq 0. \]
Consequently 
\[ x + \sum_{k=1}^{N} \epsilon_k I_n \otimes P_k^N + \sum_{k=1}^{N} 2^{k-1} t_k I_n \otimes Q_k^N \geq 0. \]
We conclude that $x \in \mathcal{C}(p_1, p_2, \dots, p_N)_n$.
\end{proof}

We will make use of the following result from \cite{araiza2020abstract}.

\begin{lemma}[{\cite[Theorem 5.3]{araiza2020abstract}}] \label{lem: projection characterization from first paper}
Let $\mathcal{V} \subseteq B(H)$ be an operator system. Suppose that $p \in \mathcal{V}$ is a projection on $H$. Then for each $x \in \mathcal{V}$ we have that $x \geq 0$ if and only if for every $\epsilon > 0$ there exists $t > 0$ such that $J_2 \otimes x + \epsilon (p \oplus p^{\perp}) + t(p^{\perp} \oplus p) \geq 0$.
\end{lemma}

\begin{proposition} \label{prop: multiple projection characterization concrete case}
Let $N \in \N$ and $\mathcal{V} \subseteq B(H)$ be an operator system and suppose that $p_1, p_2, \dots p_N \in \mathcal{V}$ are all projections on $H$. For $x \in M_n(\mathcal V)$ the following are equivalent:
\begin{enumerate}
    \item $x \in \mathcal C_n.$
    \item $x \otimes J_{2^N} \in \mathcal{C}(p_1, p_2, \dots, p_N)_n$.
    \item $x \otimes J_{2^N} \in \widehat{\mathcal{C}}(p_1, p_2, \dots, p_N)_n$.
\end{enumerate}
\end{proposition}

\begin{proof}
To see that $(1) \implies (3) \implies (2)$ notice that if $x \in \mathcal C_n$, then $x \otimes J_{2^N} \in \widehat{\mathcal C}(p_1,\dotsc,p_N)_n$ since $\mathcal C_n \otimes J_{2^N} \subseteq \widehat{\mathcal C}(p_1,\dotsc,p_N)_n$, which by Lemma~\ref{lem: p-hats contained in p's} implies $x \otimes J_{2^N} \in {\mathcal C}(p_1,\dotsc,p_N)_n.$

We will prove $(2) \implies (1)$ by induction on $N$. The case $N=1$ appears in Lemma \ref{lem: projection characterization from first paper}. Suppose that the statement holds for $N-1$, and consider the statement for $N$.

Suppose that $x \otimes J_{2^N} \in \mathcal{C}(p_1, p_2, \dots, p_N)_n$. Then for every $\epsilon_1, \epsilon_2, \dots, \epsilon_N >0$ there exist $t_1,$ $ t_2,\dotsc, t_N >0 $ such that
\[ x \otimes J_{2^N} + \sum_{k=1}^{N-1} \epsilon_k I_n \otimes P_k^N + \epsilon_N I_{n2^{N-1}} \otimes (p_N \oplus p_N^{\perp}) + \sum_{k=1}^{N-1} t_k I_n \otimes Q_k^N + t_N I_{n2^{N-1}} \otimes (p_N^{\perp} \oplus p_N)  \in \mathcal{C}_{n2^N}. \]
Applying the canonical shuffle map $\vp: M_n \otimes M_{2^{N-1}} \otimes M_2 \to M_2 \otimes M_n \otimes M_{2^{N-1}}$, we obtain
\[ J_2 \otimes \left[ x \otimes J_{2^{N-1}} + \sum_{k=1}^{N-1} \epsilon_k I_n \otimes P_k^{N-1} + \sum_{k=1}^{N-1} t_k I_n \otimes Q_k^{N-1} \right] + \epsilon_N (p_N \oplus p_N^{\perp}) \otimes I_{n2^{N-1}} \]\[+ t_N (p_N^{\perp} \oplus p_N) \otimes I_{n2^{N-1}} \in \mathcal{C}_{n2^N}. \]
Since $p_N \otimes I_{n2^{N-1}}$ is a projection in $B(H^{n2^{N-1}})$, we conclude  
\[x \otimes J_{2^{N-1}} + \sum_{k=1}^{N-1} \epsilon_k I_n \otimes P_k^{N-1} + \sum_{k=1}^{N-1} t_k I_n \otimes Q_k^{N-1} \in \mathcal{C}_{n2^{N-1}}. \]
By the inductive hypothesis, $x \in \mathcal{C}_n$, and thus $(1)$ holds.
\end{proof}

Our goal is to use Proposition \ref{prop: multiple projection characterization concrete case} to abstractly characterize finite sets of projections in an operator system and ultimately to abstractly characterize finite sets of projections in an AOU space. We begin by showing that we can endow a quotient of $M_{2^N}(\mathcal{V})$ with a useful operator system structure.

\begin{lemma}\label{lem: cone multiple contractions is a matrix ordering}
Let $(\mathcal{V}, \mathcal{C}, e)$ be an operator system and suppose $p_1, p_2, \dots, p_N \in \mathcal{V}$ are positive contractions. Then $\mathcal{C}(p_1, p_2, \dots, p_N)$ is a (non-proper) matrix ordering on the vector space $M_{2^N}(\mathcal{V})$.
\end{lemma}

\begin{proof}
The proof is similar to \cite[Proposition 4.7]{araiza2020abstract}. To see that $\mathcal{C}(p_1, \dots, p_N)$ is a matrix ordering, it suffices to check that it is closed under direct sums and conjugation by scalar matrices of the form $\alpha \otimes I_{2^N}$ (since the base vector space is $M_{2^N}(\mathcal V)$). Suppose $x \in \mathcal{C}(p_1, \dots, p_N)_n$ and $y \in \mathcal{C}(p_1, \dots, p_N)_m$. If $\epsilon_1, \dots, \epsilon_N > 0$, then there exist $t_1, \dotsc, t_N > 0$ and $r_1, \dotsc, r_N > 0$ such that \[ x + \sum_{k=1}^N \epsilon_k I_n \otimes P_i^N + \sum_{k=1}^N t_k I_n \otimes Q_k^N \in \mathcal C_{n2^{N}} \] and \[ y + \sum_{k=1}^N \epsilon_k I_m \otimes P_k^N + \sum_{k=1}^N r_k I_m \otimes Q_k^N \in \mathcal C_{m2^{N}}. \] Let $s_k=\max(t_k,r_k)$ for each $k$. Then \[ x \oplus y + \sum_{k=1}^N \epsilon_k I_{n+m} \otimes P_k^N + \sum_{k=1}^N s_k I_{n+m} \otimes Q_k^N \in \mathcal C_{(n+m)2^N}. \] It follows that $x \oplus y \in \mathcal{C}(p_1, \dots, p_N)_{n+m}$.

Next let $x \in \mathcal{C}(p_1, \dots, p_N)_n$, and suppose that $\alpha \in M_{n,k}$ with $\alpha \neq 0$. Let $\epsilon_1, \dots, \epsilon_N > 0$. Then there exists $t_1, \dotsc t_N > 0$ such that \[ x + \sum_{i=1}^N \frac{\epsilon_i}{\|\alpha\|^2} I_n \otimes P_i^N + \sum_{i=1}^N t_i I_n \otimes Q_i^N \in \mathcal C_{n2^N}. \]
Conjugating this expression by $\alpha \otimes I_{2^N}$, we obtain
\[ (\alpha \otimes I_{2^N})^* x (\alpha \otimes I_{2^N}) + \sum_{i=1}^N \frac{\epsilon_i}{\|\alpha\|^2} \alpha^* \alpha \otimes P_i^N + \sum_{i=1}^N t_i \alpha^* \alpha \otimes Q_i^N \in \mathcal C_{k2^N}. \] Since $\alpha^* \alpha \leq \|\alpha\|^2 I_k$ we have
\[ (\alpha \otimes I_{2^N})^* x (\alpha \otimes I_{2^N}) + \sum_{i=1}^N \epsilon_k I_k \otimes P_i^N + \sum_{i=1}^N \|\alpha\|^2 t_i I_k \otimes Q_i^N \in \mathcal C_{k2^N}. \] We conclude that $(\alpha \otimes I_{2^N})^* x (\alpha \otimes I_{2^N}) \in \mathcal{C}(p_1, \dots, p_N)_k$. This also holds trivially when $\alpha=0$, so the proof is complete. \end{proof}

\begin{lemma}\label{lem: I 2 to the N tensor e is Archimedean matrix order unit for matrix space multiple contractions} 
Fix $N \in \mathbb{N}$ and suppose that $(\mathcal{V}, \mathcal{C}, e)$ is an operator system and that $p_1, p_2, \dots, p_N \in \mathcal{V}$ are positive contractions. Then $I_{2^N} \otimes e$ is an Archimedean matrix-order unit for the non-proper matrix-ordered vector space $(M_{2^N}(\mathcal{V}),\mathcal{C}(p_1, \dots, p_N))$.
\end{lemma}

\begin{proof}
To see that $I_{2^N} \otimes e$ is a matrix-order unit, it suffices to show that $I_{2^N} \otimes e$ is an order unit for $M_{2^N}(\mathcal{V})$ with respect to the cone $\mathcal{C}(p_1, \dots, p_N)_1$ (see, e.g., \cite[Proposition 2.4]{araiza2020abstract}). To this end, suppose that $x \in M_{2^N}(\mathcal{V})$ and that $x = x^*$. Then because $e$ is a matrix-order unit for $\mathcal{V}$, there exists $t > 0$ such that $x + t I_{2^N} \otimes e \in \mathcal{C}_{2^N} \subseteq \mathcal{C}(p_1, \dots, p_N)_1$. Thus $I_{2^N} \otimes e$ is an order unit.

We now verify the Archimedean property. Let $n \in \mathbb{N}$ with $x \in M_{n2^N}(\mathcal V)$ and suppose that $x + \epsilon I_{n2^N} \otimes e \in \mathcal{C}(p_1, \dots, p_N)_n$ for every $\epsilon > 0$. Then for every $\epsilon_1, \epsilon_2, \dots, \epsilon_N > 0$ there exists $t_1, t_2, \dotsc, t_N > 0$ such that
\begin{equation} x + \frac{\epsilon_N}{2} I_{n2^N} \otimes e + \sum_{k=1}^{N-1} \epsilon_k I_n \otimes P_k^N + \frac{\epsilon_N}{2} I_n \otimes P_N^N + \sum_{k=1}^N t_k I_n \otimes Q_k^N \in \mathcal{C}_{n2^N}. \end{equation}
Observe that
\[ I_{n2^N} \otimes e = I_n \otimes I_{2^{N-1}} \otimes (e \oplus e) = I_n \otimes P_N^N + I_n \otimes Q_N^N. \]
Thus
\[ x + \sum_{k=1}^N \epsilon_k I_n \otimes P_k^N + \sum_{k=1}^{N-1} t_k I_n \otimes Q_k^N + \left( t_N + \frac{\epsilon_N}{2}\right) I_n \otimes Q_N^N \in \mathcal{C}_{n2^N}. \]
It follows that $x \in \mathcal{C}(p_1, \dots, p_N)_n$.
\end{proof}

\begin{proposition}\label{prop: quotient of multiple contraction cone is operator system}
Let $N \in \mathbb{N}$. Suppose that $(\mathcal{V}, \mathcal{C}, e)$ is an operator system and $p_1, p_2, \dots, p_N \in \mathcal{V}$ are positive contractions.
Let $\mathcal{J} = \Span \mathcal{C}(p_1, \dots, p_N)_1 \cap - \mathcal{C}(p_1, \dots, p_N)_1$. Then 
\[ (M_{2^N}(\mathcal{V})/\mathcal{J}, \{ \mathcal{C}(p_1, \dots, p_N)_n + M_n(\mathcal{J})\}_{n=1}^\infty, I_{2^N} \otimes e + \mathcal{J}) \]
is an operator system.
\end{proposition}

\begin{proof}
The result follows from Lemma~\ref{lem: cone multiple contractions is a matrix ordering} and Lemma~\ref{lem: I 2 to the N tensor e is Archimedean matrix order unit for matrix space multiple contractions}. 
\end{proof}

\begin{lemma} \label{lem: unital embedding into quotient}
Let $N \in \mathbb{N}$ and suppose that $(\mathcal{V}, \mathcal{C}, e)$ is an operator system and $p_1, p_2, \dots, p_N \in \mathcal{V}$ are positive contractions.
Let $\mathcal{J} = \Span \mathcal{C}(p_1, \dots, p_N)_1 \cap - \mathcal{C}(p_1, \dots, p_N)_1$. Then the mapping $x \mapsto x \otimes J_{2^N} + \mathcal{J}$ from $\mathcal V$ to $M_{2^N}(\mathcal V)/ \mathcal{J}$ is unital. Moreover 
\[e \otimes J_{2^N} + \mathcal J = I_2 \otimes e \otimes J_{2^{N-1}} + \mathcal J = I_{2^2} \otimes e \otimes J_{2^{N-2}} + \mathcal J  = \cdots = I_{2^N} \otimes e + \mathcal J.
\] 
\end{lemma}

\begin{proof}
The proof is similar to \cite[Lemma 5.6]{araiza2020abstract}. We will show 
\begin{equation} I_{2^{i-1}} \otimes e \otimes J_{2^{N-i+1}} + \mathcal{J} = P_i^N + \mathcal{J},  \qquad \text{for}\,\, i=1,\dotsc,N, \label{eq: unitality equation 1} \end{equation} and
\begin{equation}
    I_{2^i} \otimes e \otimes J_{2^{N-i}} + \mathcal{J}=P_i^N + \mathcal{J}, \qquad \text{for}\,\, i=1,\dotsc,N. \label{eq: unitality eq 2}
\end{equation}

We will first show Equation~\eqref{eq: unitality equation 1}. Fix $i \in \{1,2,\dots,N\}$. Since $e = p_i + p_i^\perp$, we see that
\[ I_{2^{i-1}} \otimes e \otimes J_{2^{N-i+1}} + \mathcal{J} = (I_{2^{i-1}} \otimes p_i \otimes J_{2^{N-i+1}} + \mathcal{J}) + (I_{2^{i-1}} \otimes p_i^\perp \otimes J_{2^{N-i+1}} + \mathcal{J}). \]
We will show 
\begin{equation} I_{2^{i-1}} \otimes p_i \otimes J_{2^{N-i+1}} + \mathcal{J} = I_{2^{i-1}} \otimes (p_i \oplus 0) \otimes J_{2^{N-i}} + \mathcal{J} \label{eq: unitality eq 3}  \end{equation}
and
\begin{equation} I_{2^{i-1}} \otimes p_i^\perp \otimes J_{2^{N-i+1}} + \mathcal{J} = I_{2^{i-1}} \otimes (0 \oplus p_i^\perp) \otimes J_{2^{N-i}} + \mathcal{J}. \label{eq: unitality eq 4} \end{equation}
Equation~\eqref{eq: unitality equation 1} can then be obtained by summing Equation~\eqref{eq: unitality eq 3} and Equation~\eqref{eq: unitality eq 4}.

In order to prove Equation~\eqref{eq: unitality eq 3} we rewrite the left hand side as
\[ I_{2^{i-1}} \otimes \begin{pmatrix} p_i & p_i \\ p_i & p_i \end{pmatrix} \otimes J_{2^{N-i}} + \mathcal{J} \]
and the statement will follow by showing 
\[ \pm I_{2^{i-1}} \otimes \begin{pmatrix} 0 & p_i \\ p_i & p_i \end{pmatrix} \otimes J_{2^{N-i}}  \in \mathcal{C}(p_1, \dots, p_N)_1. \]

To prove this, let $\epsilon_1, \dotsc, \epsilon_N > 0$. Set $t_k = \epsilon_k$ for $k \neq i$. Let $t_i = 1 + \frac{1}{\epsilon_i}$. Since
\begin{eqnarray}
 \begin{pmatrix} 0 & p_i \\ p_i & p_i \end{pmatrix} + \epsilon_i \begin{pmatrix} p_i & 0 \\ 0 & p_i^\perp \end{pmatrix} + t_i \begin{pmatrix} p_i^\perp & 0 \\ 0 & p_i \end{pmatrix} & = &  \begin{pmatrix} \epsilon_i p_i & p_i \\ p_i & (1+t_i) p_i \end{pmatrix} + \begin{pmatrix} t_i p_i^\perp & 0 \\ 0 & \epsilon_i p_i^\perp \end{pmatrix} \nonumber \\ 
& = &  \begin{pmatrix} \epsilon_i & 1 \\ 1 & 2+\frac{1}{\epsilon_i} \end{pmatrix} \otimes p_i + \begin{pmatrix} 1+\frac{1}{\epsilon_i} & 0 \\ 0 & \epsilon_i \end{pmatrix} \otimes p_i^\perp \nonumber \\ & \in &  \mathcal C_{2} \nonumber 
\end{eqnarray}
and
\begin{eqnarray}
 -\begin{pmatrix} 0 & p_i \\ p_i & p_i \end{pmatrix} + \epsilon_i \begin{pmatrix} p_i & 0 \\ 0 & p_i^\perp \end{pmatrix} + t_i \begin{pmatrix} p_i^\perp & 0 \\ 0 & p_i \end{pmatrix}
& = & \begin{pmatrix} \epsilon_i p_i & -p_i \\ -p_i & (t_i-1) p_i \end{pmatrix} + \begin{pmatrix} t_i p_i^\perp & 0 \\ 0 & \epsilon_i p_i^\perp \end{pmatrix} \nonumber \\ 
& = & \begin{pmatrix} \epsilon_i & -1 \\ -1 & \frac{1}{\epsilon_i} \end{pmatrix} \otimes p_i + \begin{pmatrix} 1+\frac{1}{\epsilon_i} & 0 \\ 0 & \epsilon_i \end{pmatrix} \otimes p_i^\perp \nonumber \\ & \in &  \mathcal C_{2} \nonumber 
\end{eqnarray}
we conclude that
\[ \pm I_{2^{i-1}} \otimes \begin{pmatrix} 0 & p_i \\ p_i & p_i \end{pmatrix} \otimes J_{2^{N-i}} + \epsilon_i P_i^N + t_i Q_i^N \in \mathcal C_{2^N} \]
and hence
\[ \pm I_{2^{i-1}} \otimes \begin{pmatrix} 0 & p_i \\ p_i & p_i \end{pmatrix} \otimes J_{2^{N-i}} + \sum_{k=1}^N \epsilon_k P_k^N + \sum_{k=1}^N t_k Q_k^N \in \mathcal C_{2^N}. \]
Therefore
\[ \pm I_{2^{i-1}} \otimes \begin{pmatrix} 0 & p_i \\ p_i & p_i \end{pmatrix} \otimes J_{2^{N-i}} \in \mathcal{C}(p_1, \dots, p_N)_1, \] which proves Equation~\eqref{eq: unitality eq 3}.
A similar argument proves
\[ \pm I_{2^{i-1}} \otimes \begin{pmatrix} p_i^\perp & p_i^\perp \\ p_i^\perp & 0 \end{pmatrix} \otimes J_{2^{N-i}} \in \mathcal{C}(p_1, \dots, p_N)_1, \]
which proves Equation~\eqref{eq: unitality eq 4}.

Finally we show Equation~\eqref{eq: unitality eq 2}. Let $\epsilon_1, \dots, \epsilon_N > 0$. For $k \neq i$, set $t_k = \epsilon_k$, and set $t_i = 1$. Then
\[ \pm Q_i^N + \sum_k \epsilon_k P_k^N + \sum_k t_k Q_k^N \in \mathcal C_{2^N}. \]
Thus $\pm Q_i^N \in \mathcal{C}(p_1, \dots, p_N)_1$. It follows that 
\[ P_i^N + \mathcal{J} = (P_i^N + Q_i^N) + \mathcal J = I_{2^i} \otimes e \otimes J_{2^{N-i}} + \mathcal{J},\] which proves Equation~\eqref{eq: unitality eq 2}.
\end{proof}

We now characterize when one has multiple projections in the operator system case. We point out that when $N=1$ this is done using the methods of \cite{araiza2020abstract}.

\begin{theorem} \label{thm: Multi projection operator system characterization}
Suppose that $(\mathcal{V}, \mathcal{C}, e)$ is an operator system and $p_1, p_2, \dots, p_N \in \mathcal{V}$ are positive contractions.
The following are equivalent:
\begin{itemize}
    \item[(1)] Each $p_1,\dotsc,p_N$ is an abstract projection in $\mathcal V.$
    \item[(2)] For each $1 \leq i \leq N$ the map $\pi_{p_i}: \mathcal V \to M_2(\mathcal V)/\mathcal{J}_i$, where $\mathcal J_i:= \Span \mathcal C(p_i)_1 \cap -\mathcal C(p_i)_1$, is a complete order embedding.
    \item[(3)] The map $x \mapsto x \otimes J_{2^N} + \mathcal{J}$, where $\mathcal J := \Span \mathcal C(p_1,\dotsc,p_N)_1 \cap -\mathcal C(p_1,\dotsc,p_N)_1$, is a complete order embedding from $\mathcal V$ to \[(M_{2^N}(\mathcal{V})/\mathcal{J}, \mathcal{C}(p_1, \dots, p_N) + \mathcal J, e \otimes I_{2^N} + \mathcal{J}). \] 
\end{itemize}
\end{theorem}

\begin{proof}
The equivalence of (1) and (2) follows from Theorem~\ref{thm: main thm AR1.0}. Proposition~\ref{prop: multiple projection characterization concrete case} shows (1) implies (3). It remains to prove (3) implies (2).

Suppose $x \mapsto x \otimes J_{2^N} + \mathcal{J}$ is a complete order embedding. Fix $i \in \{1, 2, \dots, N\}$. Since $\pi_{p_i}$ is unital completely positive we need only show that $\pi_{p_i}^{-1}: \pi_{p_i}(\mathcal V) \to \mathcal V$ is completely positive. Let $x \in M_n(\mathcal V)$ such that $(\pi_{p_i})_n(x) \in \mathcal C(p_i)_n + M_n(\mathcal J_i).$ Then for every $\epsilon_i > 0$ there exists $t_i > 0$ such that
\[ x \otimes J_2 + \epsilon_i I_n \otimes (p_i \oplus p_i^{\perp}) + t_i I_n \otimes (p_i^{\perp} \oplus p_i) \in \mathcal C_{2n}. \]
By tensoring on the left by the positive matrix $J_{2^{N-i}}$ and tensoring on the right by the positive matrix $J_{2^{i-1}}$ we obtain
\[ J_{2^{N-i}} \otimes (x \otimes J_2) \otimes J_{2^{i-1}} + \epsilon_i J_{2^{N-i}} \otimes I_n \otimes (p_i \oplus p_i^{\perp}) \otimes J_{2^{i-1}} + t_i J_{2^{N-i}} \otimes I_n \otimes (p_i^{\perp} \oplus p_i) \otimes J_{2^{i-1}} \in \mathcal C_{n2^N}. \]
Applying the canonical shuffle 
\[ \varphi: M_{2^{N-i}} \otimes M_n \otimes M_2 \otimes M_{2^{i-1}} \to M_n \otimes M_{2^{N-i}} \otimes M_2 \otimes M_{2^{i-1}} \]
we have
\[ x \otimes J_{2^N} + \epsilon_i I_n \otimes \widehat{P}_i^N + t_i \widehat{Q}_i^N \in \mathcal C_{n2^N}, \] where $\widehat{P}_i^N$ and $\widehat{Q}_i^N$ are from Definition~\ref{defn: multiple projection compression cones}.
Let $\epsilon_k > 0$ for each $k \neq i$ and set $t_k = \epsilon_k$ for each $k \neq i$. Then
\[ x \otimes J_{2^N} + \sum_k \epsilon_k I_n \otimes \hat{P}_k^N + \sum_k t_k I_n \otimes \hat{Q}_k^N \in \mathcal C_{n2^N}. \]
Hence $x \otimes J_{2^N} \in \widehat{\mathcal{C}}(p_1, \dots, p_N)_n$, which by Lemma~\ref{lem: p-hats contained in p's} implies $x \otimes J_{2^N} \in \mathcal{C}(p_1, \dots, p_N)_n.$ Since $x \mapsto x \otimes J_{2^N} + \mathcal{J}$ is a complete order embedding, we conclude that $x \in \mathcal C_n$. 
\end{proof}

Our next goal is to generalize Theorem~\ref{thm: characterize projections in AOU space} to the multiple projection case; i.e., to characterize AOU spaces with multiple projections. In the single projection case, this was achieved using Lemma~\ref{lem: compression cones equal AOU case}, which essentially says we can enlarge an arbitrary operator system structure $(\mathcal V,\mathcal{C},e)$ by restricting the operator system structure \[ (M_2(\mathcal{V})/\mathcal{J}, \{\mathcal{C}(p_1)_n + M_n(\mathcal{J})\}_{n=1}^\infty, e+\mathcal{J}) \] to the subsystem $\mathcal{V} \otimes J_2 + \mathcal{J}$ identified with $\mathcal{V}$ via $x \mapsto x \otimes J_2 + \mathcal{J}$.

\begin{definition} \label{defn: multi multi compression cones}
Let $(\mathcal{V}, \mathcal{C}, e)$ be an operator system and let $p_1, p_2, \dots, p_N \in \mathcal{V}$ be positive contractions. Define $P_i^N$ and $Q_i^N$ as in Definition \ref{defn: multiple projection compression cones}. For each $L \in \N$ we define
\[ P_{i,k}^{N,L} := I_{2^{N(k-1)}} \otimes P_i^N \otimes J_{2^{N(L-k)}}, \qquad 1 \leq i \leq N \,\,\text{and}\,\, 1 \leq k \leq L. \] Analogously, we define $Q_{i,k}^{N,L}$ by 
\[ Q_{i,k}^{N,L} := I_{2^{N(k-1)}} \otimes Q_i^N \otimes J_{2^{N(L-k)}}, \qquad 1 \leq i \leq N \,\,\text{and}\,\, 1 \leq k \leq L. \]
For each $L \in \N$ we define the non-proper matrix ordering 
\[ \mathcal{C}(p_1, \dots, p_N)^L := \mathcal{C}(\underbrace{p_1, \dots, p_N, p_1, \dots, p_N, \dots, p_1, \dots, p_N}_{L\text{-times}}). \]
More specifically, $x \in \mathcal{C}(p_1, \dots, p_N)_n^L$ if and only if $x \in M_{n2^{NL}}(\mathcal{V})$, $x=x^*,$ and for each $N \times L$ matrix $(\epsilon_{i,j})$ of strictly positive real numbers there exists a corresponding $N \times L$ matrix $(t_{ij})$ of strictly positive real numbers such that
\begin{equation} x + \sum_{i,j} \epsilon_{ij} I_n \otimes P_{ij}^{N,L} + \sum_{ij} t_{ij} I_n \otimes Q_{ij}^{N,L} \in \mathcal C_{n2^{NL}}, \label{eq: cones dependent on L-tuples of multiple projections} \end{equation} with the property that if one replaces $\epsilon_{ij}$ with  $\epsilon_{ij}^\prime < \epsilon_{ij}$ then there exists $t_{ij}^\prime > t_{ij}$ such that Equation~\eqref{eq: cones dependent on L-tuples of multiple projections} still holds.
\end{definition}

Since the cones $\mathcal{C}(p_1, \dots, p_N)_n^L$ are simply the cones one obtains from Definition \ref{defn: multiple projection compression cones} when the projections $p_1, \dots, p_N$ are repeated sequentially $L$ times, the following proposition is an immediate consequence of Proposition~\ref{prop: quotient of multiple contraction cone is operator system} and Theorem~\ref{thm: Multi projection operator system characterization}.

\begin{proposition} \label{prop: Properties of multi multi compression systems}
Let $(\mathcal{V},\mathcal{C},e)$ be an operator system, let $p_1,\dots,p_N \in \mathcal{V}$ be positive contractions and let $L \in \mathbb N$. Then the following statements hold.
\begin{enumerate}
    \item Setting $\mathcal{J}^L = \Span \mathcal{C}(p_1, \dots, p_N)_1^L \cap -\mathcal{C}(p_1, \dots, p_N)_1^L$, the triple 
    \[ (M_{2^{NL}}(\mathcal{V}) / \mathcal{J}^L, \mathcal{C}(p_1, \dots, p_N)^L + \mathcal{J}^L, e + \mathcal{J}^L) \] is an operator system, and the mapping $x \mapsto x \otimes J_{2^{NL}} + \mathcal{J}^L$ is unital.
    \item The operators $p_1, \dots p_N$ are each abstract projections in $(\mathcal{V}, \mathcal{C}, e)$ if and only if the mapping $x \mapsto x \otimes J_{2^{NL}} + \mathcal{J}^L$ from $\mathcal{V}$ to $M_{2^{NL}}(\mathcal{V}) / \mathcal{J}^L$ is a complete order embedding.
\end{enumerate}
\end{proposition}

To continue, we will need to introduce inductive limits of matrix orderings. Our definitions are a special case of the inductive limit of an operator system developed in \cite{MawhinneyTodorov2017}. For the sake of completeness, we will prove the necessary results, but they follow from the results in \cite{MawhinneyTodorov2017}.

\begin{definition} \label{defn: union of matrix cones}
Let $\mathcal V$ be a $*$-vector space and suppose that $e$ is a nonzero hermitian element of $\mathcal{V}$. Suppose that for every positive integer $L$ there exists a matrix ordering $\mathcal{C}^L = \{\mathcal C_n^L\}_{n \in \N}$ such that $(\mathcal V,\mathcal{C}^L,e)$ is an operator system. We define the matrix orderings $\{\mathcal C^L\}_{L=1}^\infty$ to be \emph{nested increasing} if for each $n \in \N$ we have $\mathcal{C}_n^L \subseteq \mathcal{C}_n^{L+1}$ for all $L \in \N$. 

For a nested increasing sequence $\{\mathcal C^L\}_{L=1}^\infty$ define $\mathcal{C}_n^\infty$ to be the Archimedean closure of $\bigcup_{L=1}^\infty \mathcal{C}_n^L$; that is, $x \in \mathcal{C}_n^\infty$ if and only if for every $\epsilon > 0$ there exists $L \in \N$ such that $x + \epsilon I_n \otimes e \in \mathcal{C}_n^L$. We call $\mathcal{C}^\infty:= \{ \mathcal{C}_n^\infty\}_{n=1}^\infty $ the \textit{inductive limit} of the sequence $\{\mathcal{C}^L\}_{L=1}^\infty$.
\end{definition}

The following Lemma should be compared with \cite[Proposition 4.10]{MawhinneyTodorov2017}.

\begin{lemma} \label{lem: union of operator systems}
Let $\mathcal V$ be a $*$-vector space and suppose that $e$ is a nonzero hermitian element of $\mathcal V$. Suppose that for every positive integer $L$ there exists a matrix ordering $\mathcal{C}^L$ such that $(\mathcal V,\mathcal{C}^L,e)$ is an operator system and $\{\mathcal C^L\}_{L=1}^\infty$ is a nested increasing sequence. Then $\mathcal{C}^\infty$ is a (possibly non-proper) matrix ordering for $\mathcal{V}$ and $e$ is an Archimedean matrix-order unit for $(\mathcal{V}, \mathcal{C}^\infty)$.
\end{lemma}

\begin{proof}
First set $\widetilde{C}_n := \bigcup_{L=1}^\infty \mathcal{C}_n^L$ for each $n \in \N$. Suppose $x \in \widetilde{C}_n$ and $y \in \widetilde{C}_m$. Then there exist $L,L' > 0$ such that $x \in \mathcal{C}_n^L$ and $y \in \mathcal{C}_m^{L'}$. Therefore $x \in \mathcal{C}_n^{\max(L,L')}$ and $y \in \mathcal{C}_m^{\max(L,L')}$. This implies $x \oplus y \in \mathcal{C}_{n+m}^{\max(L,L')}$. Therefore $x \oplus y \in \widetilde{C}_{n+m}$. It is immediate that $\{\widetilde{C}_n\}_{n =1}^\infty$ is compatible since each matrix ordering $\mathcal{C}^L$ is compatible. Let $x \in M_n(\mathcal V)_{h}$. Then for each $L > 0$ there exists $t > 0$ such that $x + tI_n \otimes e \in \mathcal{C}_n^L \subseteq \widetilde{C}_n$. Thus $e$ is a matrix order unit for $\{\widetilde{{C}}_n\}_{n=1}^\infty$. Since $\mathcal{C}_n^\infty$ is the Archimedean closure of $\widetilde{C}_n$, it follows that $\mathcal{C}^\infty$ is a (possibly non-proper) matrix ordering for $\mathcal{V}$ and $e$ is an Archimedean matrix order unit for $(\mathcal{V}, \mathcal{C}^\infty)$. 
\end{proof}

\begin{proposition} \label{prop: Union multi compression cones}
Let $(\mathcal{V},\mathcal{C},e)$ be an operator system let $p_1, \dots, p_N \in \mathcal{V}$ be positive contractions. For $L \in \N$ let $\pi_L: \mathcal{V} \to M_{2^{NL}}(\mathcal{V})$ denotes the mapping $x \mapsto x \otimes J_{2^{NL}}$. Then 
\[ \{\pi_L^{-1}(\mathcal{C}(p_1, \dots, p_N)^L)\}_{L=1}^\infty \]
is a nested increasing sequence of matrix orderings on $\mathcal V$. 
\end{proposition}

\begin{proof}
Fix $n, L \in \N$ and let $x \in M_n(\mathcal V).$ Suppose that $x \otimes J_{2^{NL}} \in \mathcal{C}(p_1, \dots, p_N)_n^L$. Hence, for every $N \times L$ matrix $(\epsilon_{i,k})$ of strictly positive real numbers there exists a corresponding $N \times L$ matrix $(t_{i,k})$ of strictly positive real numbers such that
\[ x \otimes J_{2^{NL}} + \sum_{i=1}^N \sum_{k=1}^{L} \epsilon_{i,k} I_n \otimes P_{i,k}^{N,L} + \sum_{i=1}^N \sum_{k=1}^{L} t_{i,k} I_n \otimes Q_{i,k}^{N,L} \in \mathcal{C}_{n2^{NL}}. \]
Tensoring this expression on the right by $J_{2^{N}}$ we obtain
\[ x \otimes J_{2^{N(L+1)}} + \sum_{i=1}^N \sum_{k=1}^{L}\epsilon_{i,k} I_n \otimes P_{i,k}^{N,L} \otimes J_{2^N} + \sum_{i=1}^N \sum_{k=1}^{L} t_{i,k} I_n \otimes Q_{i,k}^{N,L} \otimes J_{2^N} \in \mathcal{C}_{n2^{N(L+1)}}. \]
 Let $\epsilon_{1,L+1}, \dots, \epsilon_{N,L+1} > 0$ and set $t_{i,L+1}:= \epsilon_{i,L+1}$. Then
\[ x \otimes J_{2^{N(L+1)}} + \sum_{i=1}^N \sum_{k=1}^{L+1} \epsilon_{i,k} I_n \otimes P_{i,k}^{N,L+1} + \sum_{i=1}^N \sum_{k=1}^{L+1} t_{i,k} I_n \otimes Q_{i,k}^{N,L+1} \in \mathcal{C}_{n2^{N(L+1)}}. \] 
We conclude that $x \otimes J_{2^{N(L+1)}} \in \mathcal{C}(p_1, \dots, p_N)_n^{L+1}$, and hence 
\[ (\pi_L)_n^{-1}(\mathcal{C}(p_1, \dots, p_N)_n^L) \subseteq (\pi_L)_n^{-1}(\mathcal{C}(p_1, \dots, p_N)_n^{L+1}). \qedhere \] \end{proof}

\begin{definition} \label{defn: C(p_1,...,p_n)^infty}
Let $(\mathcal{V},\mathcal{C},e)$ be an operator system and let $p_1,\dots,p_N \in \mathcal{V}$ be positive contractions. We define the matrix ordering $\mathcal{C}(p_1, \dots, p_N)^\infty$ on $\mathcal{V}$ to be the inductive limit of the nested increasing sequence $\{\pi_L^{-1}(\mathcal{C}(p_1, \dots, p_N)^L)\}_{L=1}^\infty$ on $\mathcal{V}$ where $\pi_L: \mathcal{V} \to M_{2^{NL}}(\mathcal{V})$ is given by $x \mapsto x \otimes J_{2^{NL}}$.
\end{definition}

\begin{lemma} \label{lem: Proper matrix cone}
Let $\mathcal{C}$ be a matrix ordering on a $*$-vector space $\mathcal{V}$, and suppose that $\mathcal{C}_1$ is proper. Then $\mathcal{C}_n$ is proper for every $n \in \N$.
\end{lemma}

\begin{proof}
Suppose that $x = (x_{ij}) \in \mathcal{C}_n \cap -\mathcal{C}_n$. Then for each $k= 1, 2, \dots, n$ we have $x_{kk} = e_k^* x e_k \in \mathcal{C}_1 \cap -\mathcal{C}_1$, where $e_k \in \mathbb{C}^n$ is the $k$\textsuperscript{th} standard unit vector. Hence $x_{kk} = 0$ since $\mathcal C_1$ is proper. Let $k,l \in \{1,2, \dots, n\}$ with $k \neq l$. Then $(e_k + e_l)^* x (e_k + e_l) = x_{lk} + x_{kl} = 2\operatorname{Re}(x_{lk}) \in \mathcal{C}_1 \cap - \mathcal{C}_1$. Hence $\operatorname{Re}(x_{lk}) = 0$. Also $(e_k - ie_l)^* x (e_k - ie_l) = i(x_{lk} - x_{kl}) = 2i\operatorname{Im}(x_{lk}) \in \mathcal{C}_1 \cap - \mathcal{C}_1$. Hence $\operatorname{Im}(x_{lk}) = 0$. Thus $x_{lk} = 0$, and it follows that $x = 0$ and $\mathcal{C}_n$ is proper.
\end{proof}

\begin{proposition} \label{prop: restrict union of multi compression cones}
Let $(\mathcal{V},\mathcal{C},e)$ be an operator system and let $p_1,\dots,p_N \in \mathcal{V}$ be positive contractions. If the cone $\mathcal{C}(p_1, \dots, p_N)_1^\infty$ is proper then $(\mathcal{V}, \mathcal{C}(p_1, \dots, p_N)^\infty, e)$ is an operator system, and for each $k=1,\dotsc,N$ we have that $p_k$ is an abstract projection in $(\mathcal{V}, \mathcal{C}(p_1, \dots, p_N)^\infty, e)$.
\end{proposition}

\begin{proof}
Since $\mathcal{V}$ is an operator system and $p_1, p_2, \dots, p_N \in \mathcal{V}$ are contractions, it follows from the first statement of Proposition~\ref{prop: Properties of multi multi compression systems} that $\pi_L^{-1}(\mathcal{C}(p_1, \dots, p_N)^L)$ is a matrix ordering on $\mathcal{V}$ for each $L \in \N$ and that $e$ is an Archimedean matrix order unit for each matrix ordered vector space $(\mathcal{V}, \pi_L^{-1}(\mathcal{C}(p_1, \dots, p_N)^L))$. Then by Lemma~\ref{lem: union of operator systems} and Definition~\ref{defn: C(p_1,...,p_n)^infty} we conclude that the inductive limit $\mathcal{C}(p_1, \dots, p_N)^\infty$ is a matrix ordering on $\mathcal{V}$ and $e$ is an Archimedean matrix order unit for the matrix ordered $*$-vector space $(\mathcal{V}, \mathcal{C}(p_1, \dots, p_N)^\infty)$. If $\mathcal{C}(p_1, \dots, p_N)_1^\infty$ is proper, by Lemma~\ref{lem: Proper matrix cone} $\mathcal{C}(p_1, \dots, p_N)^\infty$ is a proper matrix ordering. Therefore $(\mathcal{V}, \mathcal{C}(p_1, \dots, p_N)^\infty, e)$ is an operator system.

Given $x \in M_n(\mathcal V)$ suppose that for every $\epsilon_1, \dots, \epsilon_N > 0$ there exists $t_1, \dots, t_N > 0$ such that 
\[ x \otimes J_{2^N} + \sum_i \epsilon_i I_n \otimes P_i^N + \sum_i t_i I_n \otimes Q_i^N \in \mathcal{C}(p_1, \dots, p_N)_{n2^N}^\infty. \]
Then for every $\epsilon > 0$ there exists $L \in \N$ such that
\[ \left[ x \otimes J_{2^N} + \sum_i \epsilon_i I_n \otimes P_i^N + \sum_i t_i I_n \otimes Q_i^N + \epsilon I_{n2^N} \otimes e \right] \otimes J_{2^{NL}} \in \mathcal{C}(p_1, \dots, p_N)_{n2^N}^L, \] and in particular, \[
\left[ x \otimes J_{2^N} + \sum_i \epsilon_i I_n \otimes P_i^N + \sum_i t_i I_n \otimes Q_i^N + \epsilon I_{n2^N} \otimes e \right] \otimes J_{2^{NL}} + M_{n2^N}(\mathcal{J}^L) \] \[ \in \mathcal{C}(p_1, \dots, p_N)_{n2^N}^L + M_{n2^{NL}}(\mathcal J^L).
\]
By Lemma~\ref{lem: unital embedding into quotient} we see that $I_{2^N} \otimes e \otimes J_{2^{NL}} + M_{2^N}(\mathcal{J}^L) = e \otimes J_{2^{N(L+1)}} + M_{2^N}(\mathcal{J}^L)$. Hence \[I_{n2^N} \otimes e \otimes J_{2^{NL}} + M_{n2^N}(\mathcal{J}^L) = I_n \otimes e \otimes J_{2^{N(L+1)}} + M_{n2^N}(\mathcal{J}^L).\] This implies
\[ \left[ x \otimes J_{2^N} + \sum_i \epsilon_i I_n \otimes P_i^N + \sum_i t_i I_n \otimes Q_i^N + \epsilon I_n \otimes e \otimes J_{2^N} \right] \otimes J_{2^{NL}} \in \mathcal{C}(p_1, \dots, p_N)_{n2^N}^L. \]
Therefore for every $N \times L$ matrix $(\delta_{i,k})$ of strictly positive real numbers there exists an $N \times L$ matrix $(r_{i,k})$ of strictly positive real numbers such that
\[ \left[ x \otimes J_{2^N} + \sum_i \epsilon_i I_n \otimes P_i^N + \sum_i t_i I_n \otimes Q_i^N + \epsilon I_n \otimes  e \otimes J_{2^N} \right] \otimes J_{2^{NL}} \]\[+ \sum_{i,k} \delta_{i,k} I_{n2^N} \otimes P_{i,k}^{N,L} + \sum_{i,k} r_{ik} I_{n2^N} \otimes Q_{i,k}^{N,L} \in \mathcal C_{n2^{N(L+1)}}.\] 
Since $I_{n2^N} \otimes P_{i,k}^{N,L} = I_n \otimes P_{i,k+1}^{N,L+1}$, $I_{n2^N} \otimes Q_{i,k}^{N,L} = I_n \otimes Q_{i,k+1}^{N,L+1}$, $I_n \otimes P_i^N \otimes J_{2^{NL}} = I_n \otimes P_{i,1}^{N,L+1}$, and $I_n \otimes Q_i^N \otimes J_{2^{NL}} = I_n \otimes Q_{i,1}^{N,L+1}$, by setting $\epsilon_i = \delta_{i,0}$ and $t_i = r_{i,0}$ for $1 \leq i \leq N$, we conclude that
\[ (x + \epsilon I_n \otimes e) \otimes J_{2^{N(L+1)}} + \sum_{i=1}^N \sum_{k=1}^{L+1} \delta_{i,k-1} I_n \otimes P_{i,k}^{N,L+1} + \sum_{i=1}^N \sum_{k=1}^{L+1} r_{i,k-1} I_n \otimes Q_{i,k}^{N,L+1} \in \mathcal C_{n 2^{N(L+1)}}, \]
and hence
\[ (x + \epsilon I_n \otimes e) \otimes J_{2^{N(L+1)}} \in \mathcal{C}(p_1, \dots, p_N)_n^{L+1}. \]
Since $\epsilon >0$ was arbitrary, we conclude that $x \in \mathcal{D}_n := \mathcal{C}(p_1, \dots, p_N)_n^\infty$ by the Archimedean property. Therefore $x \mapsto x \otimes J_{2^N} + \mathcal{J}$ is a complete order embedding of $(\mathcal V,\mathcal{D},e)$ into the quotient operator system 
\[ (M_{2^N}(\mathcal{V})/\mathcal{J}, \mathcal{D}(p_1, \dots, p_N) + \mathcal{J}, e \otimes I_{2^N} + \mathcal{J}) \] where 
\[ \mathcal{J} = \Span \mathcal{D}(p_1, \dots, p_N) \cap - \mathcal{D}(p_1, \dots, p_N). \]
The statement then follows from Theorem~\ref{thm: Multi projection operator system characterization}.

\end{proof}

We can now characterize multiple projections in an AOU space.

\begin{theorem} \label{thm: Multi projections AOU characterization}
Let $(\mathcal V,C,e)$ be an AOU space and let $p_1, p_2, \dots, p_N$ be positive contractions in $\mathcal V$. Then the following statements are equivalent:
\begin{enumerate}
    \item There exists a Hilbert space $H$ and a unital order embedding $\pi: \mathcal V \to B(H)$ such that $\pi(p_i)$ is a projection on $H$ for each $i=1,\dots,N$.
    \item There exists a proper matrix ordering $\mathcal{C}$ with $\mathcal{C}_1=C$ such that $(\mathcal V, \mathcal{C}, e)$ is an operator system with abstract projections $p_1, p_2, \dots, p_N$.
    \item $C = C^{max}(p_1, \dots, p_N)_1^\infty$.
\end{enumerate}
\end{theorem}

\begin{proof}
The equivalence of (1) and (2) follows from Theorem \ref{thm: Multi projection operator system characterization}. We shall establish the equivalence of (2) and (3).

First suppose (2) holds. Note that $C = C_1^{max}$ where $C^{max}$ denotes the maximal matrix ordering on the AOU space $(\mathcal V, C, e).$ Since $C_1^{max} \subseteq C^{max}(p_1,\dotsc,p_N)_1^L$ for each $L \in \N$, we have $C_1^{max} \subseteq C^{max}(p_1,\dotsc,p_N)_1^\infty$, and hence $C \subseteq C^{max}(p_1,\dotsc,p_N)_1^\infty.$ For the reverse inclusion, if $x \in  C^{{max}}(p_1,\dotsc,p_N)_1^\infty$, then for every $\epsilon >0$ there exists $L \in \N$ such that $(x+\epsilon e) \otimes J_{2^{NL}} \in C^{max}(p_1,\dotsc,p_N)_1^L$. Since $ C^{max}(p_1,\dotsc,p_N)_1^L \subseteq \mathcal C(p_1,\dotsc,p_N)_1^L$, for every $\epsilon >0$ there exists $L \in \N$ such that $(x + \epsilon e) \otimes J_{2^{NL}} \in \mathcal C(p_1,\dotsc,p_N)_1^L$. Because $p_1,\dotsc,p_N$ are abstract projections relative to $\mathcal C$, it follows that $x + \epsilon e \in \mathcal C_1$. Because this holds for all $\epsilon >0$ and $\mathcal C_1$ is Archimedean closed, it follows that $x \in \mathcal C_1 = C.$ Thus $(3)$ holds.

Next assume (3) holds. Since $C =C^{max}(p_1, \dots, p_N)_1^\infty$ and $C$ is a proper cone, Proposition~\ref{prop: restrict union of multi compression cones} implies $(\mathcal V, C^{max}(p_1, \dots, p_N)^\infty ,e)$ is an operator system with abstract projections $p_1, \dots, p_N$. Thus the matrix ordering $C^{max}(p_1,\dotsc,p_N)^\infty$ satisfies the conditions of (2).

\end{proof}

The following corollary is immediate from Proposition~\ref{prop: restrict union of multi compression cones} and Theorem~\ref{thm: Multi projections AOU characterization}.

\begin{corollary} \label{cor: projection envelope of an AOU space}
Let $(\mathcal V,C,e)$ be an AOU space and let $p_1, p_2, \dots, p_N$ be positive contractions in $\mathcal V$. Suppose that the cone $C^{max}(p_1, \dots, p_N)_1^\infty$ is proper. Then there exists a unital order embedding $\pi: (\mathcal V,C^{max}(p_1,\dotsc,p_N)_1^\infty,e) \to B(H)$ such that $\pi(p_k)$ is a projection for each $k=1,\dotsc,N$. 
\end{corollary}

\begin{proof}
Since $C^{max}(p_1, \dots, p_N)_1^\infty$ is proper then by Proposition~\ref{prop: restrict union of multi compression cones} the triple \[ (\mathcal V,C^{max}(p_1,\dotsc,p_N)^\infty,e) \] is an operator system and each $p_i$ is an abstract projection in $(\mathcal V, C^{max}(p_1,\dotsc,p_N)^\infty,e)$.  In particular, $(\mathcal V, C^{max}(p_1,\dotsc,p_N)_1^\infty,e)$ is an AOU space and thus by Condition (2) and Condition (3) of Theorem~\ref{thm: Multi projections AOU characterization} there exists a unital order embedding $\pi: (\mathcal V, C^{max}(p_1,\dotsc,p_N)^\infty,e) \to B(H)$ mapping each $p_i$ to a projection. 
\end{proof}

\section{Quantum commuting correlations as states on AOU spaces}\label{sec: qc correlations as states in AOU spaces}

Let $n,k \in \mathbb{N}$. We call a tuple $p = \{p(a,b|x,y) : a,b \in [k], x,y \in [n]\}$ a \textit{correlation} with $n$ inputs and $k$ outputs if for each $a,b \in [k]$ and $x,y \in [n]$, $p(a,b|x,y)$ is a non-negative real number, and for each $x,y \in [n]$ we have \[ \sum_{a,b = 1}^k p(a,b|x,y) = 1. \] We let $C(n,k)$ denote the set of all correlations with $n$ inputs and $k$ outputs. A correlation $p$ is called \textit{nonsignalling} if for each $a,b \in [k]$ and $x,y \in [n]$ the values \[ p_A(a|x) := \sum_d p(a,d|x,w) \quad \text{ and } \quad p_B(b|y) := \sum_c p(c,b|z,y) \] are well-defined, meaning that $p_A(a|x)$ is independent of the choice of $w \in [n]$ and $p_B(b|y)$ is independent of the choice of $z \in [n]$. We let $C_{ns}(n,k)$ denote the set of all nonsignalling correlations with $n$ inputs and $k$ outputs.

Much of the literature on correlation sets is focused on various subsets of the nonsignalling correlation sets. We mention three of these subsets here, namely the quantum commuting, quantum, and local correlations. A correlation $p$ is a \textit{quantum commuting} correlation with $n$ inputs and $k$ outputs if there exists a Hilbert space $H$, a pair of C*-algebras $\mathcal{A}, \mathcal{B} \subseteq B(H)$ with $z_1z_2=z_2z_1$ for all $z_1 \in \mathcal{A}$ and $z_2 \in \mathcal{B}$, projection-valued measures $\{E_{x,a}\}_{a=1}^k \subseteq \mathcal{A}$ and $\{F_{y,b}\}_{b=1}^k \subseteq \mathcal{B}$ for each $x,y \in [n]$, and a state $\phi: \mathcal{A} \mathcal{B} \to \mathbb{C}$ such that $p(a,b|x,y) = \phi(E_{x,a} F_{y,b})$ for all $a,b \in [k]$ and $x,y \in [n]$. A quantum commuting correlation is called a \textit{quantum} correlation if we require the Hilbert space $H$ to be finite-dimensional. A quantum commuting correlation is called \textit{local} if we require that the C*-algebras $\mathcal{A}$ and $\mathcal{B}$ are commutative. We let $C_{qc}(n,k), C_{q}(n,k)$, and $C_{loc}(n,k)$ denote the sets of quantum commuting, quantum, and local correlations, respectively.

It is well-known that for each input-output pair $(n,k)$ each of the correlation sets mentioned above are convex subsets of $\mathbb{R}^{n^2 k^2}$ and satisfy
\[ C_{loc}(n,k) \subseteq C_{q}(n,k) \subseteq C_{qc}(n,k) \subseteq C_{ns}(n,k) \subseteq C(n,k). \]
Moreover, each inclusion in the above sequence is proper for some choice of input $n$ and output $k$. The set $C_q(n,k)$ is known to be non-closed for certain values of $n$ and $k$ \cite{slofstra2019set}. Let $C_{qa}(n,k)$ denote the closure of $C_{q}(n,k)$ for each pair $(n,k)$. Whether or not $C_{qa}(n,k)$ is equal to $C_{qc}(n,k)$ for every input $n$ and output $k$ is equivalent to Connes' embedding problem from operator algebras (see \cite{junge2011connes}, \cite{fritz2012tsirelson}, and \cite{ozawa2013connes}). The recent preprint \cite{ji2020mip} implies that $C_{qa}(n,k)$ does not equal $C_{qc}(n,k)$ for some values of $n$ and $k$, estimated to be approximately $10^{20}$ each. It remains open whether or not $C_{qa}(n,k)$ differs from $C_{qc}(n,k)$ for small values of $n$ and $k$. Although it is well-known that $C_q(2,2) = C_{qa}(2,2) = C_{qc}(2,2)$, it remains open whether $C_{qa}(3,2) = C_{qc}(3,2)$ or $C_{qa}(3,2) = C_{q}(3,2)$, for example.

Let \[ G(n,k) := \underbrace{(\mathbb{Z} / k \mathbb{Z}) * \dots * (\mathbb{Z} / k \mathbb{Z})}_{n\,\, \text{times}}; \] i.e., $G(n,k)$ is the $n$-fold free product of $\mathbb{Z} / k \mathbb{Z}$ with itself amalgamated over the group identity. Let $\mathcal{A}(n,k) := C^*(G(n,k))$ denote the full group C*-algebra of $G(n,k)$. Then $\mathcal{A}(n,k)$ is the universal C*-algebra generated by $n$ noncommuting projection-valued measures $\{P_{1,i}\}_{i=1}^k, \dots \{P_{n,i}\}_{i=1}^k$ (here we take $\{P_{x,i}\}_{i=1}^k$ to be the spectral projections for the generator of the $x$\textsuperscript{th} factor $\mathbb{Z} / k \mathbb{Z}$ of the free product). Let $\mathcal{S}(n,k) \subseteq \mathcal{A}(n,k)$ denote the operator system spanned by the projections $\{P_{x,i}: x \in [n], i \in [k]\}$ for  in $\mathcal A(n,k)$.

\begin{theorem}[\cite{lupini2020perfect}] \label{thm: lupini et. al.}
The following statements are equivalent:
\begin{enumerate}
    \item $p \in C_{ns}(n,k)$ (resp. $C_{qc}(n,k)$, $C_{qa}(n,k)$).
    \item There exists a state $\phi$ on $\mathcal{S}(n,k) \otimes_{max} \mathcal{S}(n,k)$ (resp. $\mathcal{S}(n,k) \otimes_c \mathcal{S}(n,k)$, $\mathcal{S}(n,k) \otimes_{min} \mathcal{S}(n,k)$) such that $p(a,b|x,y) = \phi( P_{x,a} \otimes P_{y,b})$ for each $a,b \in [k]$ and $x,y \in [n]$.
\end{enumerate}
\end{theorem}

A corollary to Theorem~\ref{thm: lupini et. al.} is that there exist AOU spaces $(\mathcal{V}_r, \mathcal{V}_r^+, e)$ for $r = ns, qc, qa$ each generated by positive operators $Q(a,b|x,y)$ such that $p(a,b|x,y) \in C_r(n,k)$ if and only if there exists a state $\phi:\mathcal{V}_r \to \mathbb{C}$ such that $p(a,b|x,y) = \phi(Q(a,b|x,y))$. Hence, determining the geometry of the positive cones $\mathcal{V}_r^+$ is equivalent to determining the geometry of the set $C_r(n,k)$ by Kadison's Theorem \cite{kadison1951representation}.

We recall the following definition from \cite{araiza2020abstract}.

\begin{definition}[{\cite[Definition 6.2]{araiza2020abstract}}] \label{defn: qc operator system}
Let $n, k \in \N$. We call an operator system $(\mathcal{V},\mathcal{C},e)$ a \textit{nonsignalling} operator system if $\mathcal{V} = \Span \{Q(a,b|x,y)\}_{a,b \in [k]; x,y \in [n]}$ where $Q(a,b|x,y) \in \mathcal C_1$ for each $a,b \in [k]$ and $x,y \in [n]$,  
\[ \sum_{a,b =1}^k Q(a,b|x,y) = e \]
for each $x,y \in [n]$, and the \textit{marginal vectors}
\[ E(a|x) := \sum_{c=1}^k Q(a,c|x,z) \quad \text{ and } \quad F(b|y) := \sum_{d=1}^k Q(d,b|w,y) \]
are well-defined. We call $\mathcal{V}$ a \textit{quantum commuting} operator system if it is a nonsignalling operator system with the additional condition that each $Q(a,b|x,y)$ is an abstract projection in $(\mathcal{V}, \mathcal{C}, e)$.
\end{definition}

The terminology in Definition \ref{defn: qc operator system} is justified by the following theorem.

\begin{theorem}[{\cite[Theorem 6.3]{araiza2020abstract}}] \label{thm: NS and QC operator systems}
A correlation $p \in C(n,k)$ is nonsignalling (resp. quantum commuting) if and only if there exists a nonsignalling (resp. quantum commuting) operator system $\mathcal{V}$ with generators $\{Q(a,b|x,y); a,b \in [k], x,y \in [n] \}$ and a state $\phi$ on $\mathcal{V}$ such that 
\[ p(a,b|x,y) = \phi(Q(a,b|x,y)) \]
for each $a,b \in [k]$ and $x,y \in [n]$.
\end{theorem}

Our goal is to construct a pair of AOU spaces universal with respect to nonsignalling and quantum commuting correlations, in the sense that $p$ is a nonsignalling (resp. quantum commuting) correlation if and only if $p(a,b|x,y) = \phi(Q(a,b|x,y))$ for some state $\phi$, where $\{Q(a,b|x,y): a,b \in [k], x,y \in [n]\}$ contains certain canonical elements of these AOU spaces. We begin by generalizing Definition \ref{defn: qc operator system}.

\begin{definition}
We call a pair $(\mathcal{V}, \{Q(a,b|x,y)\}_{a,b \in [k], x,y \in [n]})$ a \textit{nonsignalling} vector space on $n$ inputs and $k$ outputs if $\mathcal V$ is a vector space spanned by vectors $\{Q(a,b|x,y): a,b \in [k], x,y \in [n] \}$ satisfying \[ \sum_{a,b \in [k]} Q(a,b|x,y) = e \] for some fixed nonzero vector $e$, which we call the \textit{unit} of $\mathcal{V}$, and such that the vectors \[ E(a|x) := \sum_{c=1}^k Q(a,c|x,z) \quad \text{ and } \quad F(b|y) := \sum_{d=1}^k Q(d,b|w,y) \] are well-defined. When the vectors $Q(a,b|x,y)$ are clear from context, we simply call $\mathcal{V}$ a nonsignalling vector space. When $\mathcal{V}$ is nonsignalling, we write $n(\mathcal{V})$ and $k(\mathcal{V})$ for the number of inputs and for the number of outputs, respectively; i.e., $\mathcal{V} = \Span \{Q(a,b|x,y) : a,b \in [k(\mathcal{V})], x,y \in [n(\mathcal{V})] \}.$
\end{definition}

We use the notation $n(\mathcal{V})$ and $k(\mathcal{V})$ to speak of nonsignalling vector spaces without reference to the input and and output parameters $n$ and $k$, only referencing them when needed. However, it should be understood that any nonsignalling vector space $\mathcal{V}$ comes equipped with spanning vectors $\{Q(a,b|x,y)\}_{a,b \in [k], x,y \in [n]}$ and associated input and output parameters $n(\mathcal{V})$ and $k(\mathcal{V})$ whether they are mentioned or not. We remark that a vector space $\mathcal{V}$ could have two spanning sets $\{Q(a,b|x,y)\}$ and $\{P(a,b|x,y)\}$ with possibly different input and output parameters making $\mathcal{V}$ into a nonsignalling vector space. For example, if $E(a|x), F(b|y) \in \mathbb{R}$ are arbitrary positive real numbers satisfying $\sum_{a} E(a|x) = \sum_b F(b|y) = 1$ for all $x,y \in [n]$, then setting $Q(a,b|x,y)=E(a|x)F(b|y)$ makes $\mathbb{C}$ into a nonsignalling vector space.

The following example will play a role in several proofs hereafter.

\begin{example} \label{ex: commutative ns vector space}
Let $n,k \in \mathbb{N}$. Let $D_k$ denote the set of diagonal $k \times k$ matrices. Let $E_a$ denote the diagonal matrix with 1 for its $a$\textsuperscript{th} diagonal entry and zeroes elsewhere. Let $\mathcal{V} \subseteq D_k^{\otimes 2n}$ denote the vector space spanned by the operators $\{Q(a,b|x,y): a,b \in [k], x,y \in [n] \}$ defined by
\[ Q(a,b|x,y):= I_k^{\otimes x-1} \otimes E_{a} \otimes I_k^{\otimes n-x} \otimes I_k^{\otimes y-1} \otimes E_{b} \otimes I_k^{\otimes n-y} \]
where $I_k$ denotes the $k \times k$ identity matrix and $I_k^{\otimes n}$ denotes the $n$-fold tensor product of $I_k$ with itself (understanding $I_k^{\otimes 0} = 1$). Then $\mathcal{V}$ is a nonsignalling vector space with generators $\{Q(a,b|x,y)\}$. Moreover $\dim(\mathcal{V}) = (n(k-1)+1)^2$. To see this, first observe that for each $a,b \in [k]$ and $x,y \in [n]$
\[ E(a|x) = I_k^{\otimes x-1} \otimes E_{a} \otimes I_k^{\otimes n-x} \otimes I_k^{\otimes n} \quad \text{ and } \quad F(b|y) = I_k^{\otimes n} \otimes I_k^{\otimes y-1} \otimes E_{b} \otimes I_k^{\otimes n-y}. \]
Moreover $Q(a,b|x,y) = E(a|x) F(b|y)$, where the product is taken in the algebra $D_k^{\otimes 2n}$. Hence $\mathcal{V} = \mathcal{V}_A \mathcal{V}_B$ where $\mathcal{V}_A$ denotes $\Span \{E(a|x): a \in [k], x \in [n] \}$ and $\mathcal{V}_B$ denotes $\Span \{F(b|y) : b \in [k], y \in [n] \}$. It follows that $\mathcal{V}_A$ is spanned by the set 
\[ S = \{E(a|x): a \in [k-1], x \in [n] \} \cup \{ I_k^{\otimes 2n} \} \]
since $E(k|x) = I_k^{\otimes 2n} - (\sum_{a=1}^{k-1} E(a|x))$. It follows that $S$ is linearly independent and hence a basis for $\mathcal{V}_A$. Thus $\dim(\mathcal{V}_A) = n(k-1) + 1$. A similar observation shows that $\dim(\mathcal{V}_B) = n(k-1) + 1$. Hence
\[ \dim(\mathcal{V}) = \dim(\mathcal{V}_A) \dim(\mathcal{V}_B) = (n(k-1)+1)^2. \]
\end{example}

For the remainder of this section, we establish the existence and describe the structure of universal nonsignalling and quantum commuting operator systems and AOU spaces. Indeed, the existence of such spaces is already implied by Theorem \ref{thm: lupini et. al.} as remarked above, so the details of the constructions will be more important than the observation that such spaces exist. Each of these structures will have a common underlying vector space, which we describe in the next proposition.

\begin{proposition} \label{prop: universal non-signalling vector space}
For each $n,k \in \mathbb{N}$, there exists a nonsignalling vector space $\mathcal{V}_{ns}$ with $ n(\mathcal{V}_{ns})=n$, $k(\mathcal{V}_{ns})=k$, and generators $\{Q_{ns}(a,b|x,y) \}$ satisfying the following universal property: if $\mathcal{W}$ is another nonsignalling vector space with $n(\mathcal{W}_{ns}) = n$, $k(\mathcal{W}_{ns}) = k$, and generators $\{Q(a,b|x,y)\}$, then there exists a linear map $\phi: \mathcal{V}_{ns} \to \mathcal{W}$ satisfying $\phi(Q_{ns}(a,b|x,y)) = Q(a,b|x,y)$. Moreover $\dim(\mathcal{V}_{ns}) = (n(k-1)+1)^2$.
\end{proposition}

\begin{proof}
Let $\mathcal{W}$ be any nonsignalling vector space with generators $\{Q(a,b|x,y)\}$ satisfying $n(\mathcal{W}_{ns}) = n$ and $k(\mathcal{W}_{ns}) = k$. Let $\widetilde{V} := \mathbb{C}^{n^2 k^2}$ and denote the canonical basis elements as $\{ \widetilde{Q}(a,b|x,y): a,b \in [k], x,y \in [n]\}$. Let $\phi: \widetilde{V} \to \mathcal W$ be the linear map that takes $\widetilde{Q}(a,b|x,y) \mapsto Q(a,b|x,y)$. Define the vectors
\begin{eqnarray} F(x,y|x',y') & := & \sum_{a,b} \widetilde{Q}(a,b|x,y) - \sum_{a,b} \widetilde{Q}(a,b|x',y'), \nonumber \\
G(a|x,z,w) & := & \sum_c \widetilde{Q}(a,c|x,z) - \sum_c \widetilde{Q}(a,c|x,w), \quad \text{and} \nonumber \\
H(b|y,z,w) & := & \sum_d \widetilde{Q}(d,b|z,y) - \sum_d \widetilde{Q}(d,b|w,y). \nonumber
\end{eqnarray}
Let $J$ be the subspace spanned by the vectors $\{F(x,y|x',y'), G(a|x,z,w), H(b|y,z,w)\}$. Then $J \subseteq \ker \phi$. Set $\mathcal{V}_{ns} := \widetilde{V}/J$ and $Q_{ns}(a,b|x,y) := \widetilde{Q}(a,b|x,y) + J \in \mathcal{V}_{ns}$. Then $\phi$ descends to a linear map $\widetilde{\phi}: \mathcal{V}_{ns} \to \mathcal W$ taking $Q_{ns}(a,b|x,y) \mapsto Q(a,b|x,y)$. Moreover $\mathcal{V}_{ns}$ is a nonsignalling vector space with unit $e := \sum_{a,b} Q(a,b|x,y)$. Since $\widetilde{\phi}$ is linear and since $\mathcal{W}$ is arbitrary, $\mathcal{V}_{ns}$ satisfies the conditions of the theorem.

We now calculate $\dim(\mathcal{V}_{ns})$. It suffices to show that the set
\[ B = \{e_{ns}, Q_{ns}(a,b|x,y), E_{ns}(a|x), F_{ns}(b|y): a,b \in [k-1], x,y \in [n] \} \]
is a basis for $\mathcal{V}_{ns}$ (where $E_{ns}(a|x)$ and $F_{ns}(b|y)$ are the marginal vectors and $e_{ns}$ is the unit of $\mathcal{V}_{ns}$), because this set contains $n^2(k-1)^2 + 2n(k-1) + 1 = (n(k-1)+1)^2$ elements. To see that $B$ spans $\mathcal{V}_{ns}$, we only need to show that $\Span B$ contains the vectors $Q_{ns}(k,b|x,y), Q_{ns}(a,k|x,y)$, and $Q_{ns}(k,k|x,y)$ for each $a,b \in [k-1]$ and $x,y \in [n]$. To this end, fix $b \in [k]$ and $x,y \in [n]$. Then
\[ F_{ns}(b|y) -  \sum_{a=1}^{k-1} Q_{ns}(a,b|x,y)  = \sum_{a=1}^k Q_{ns}(a,b|x,y)  -  \sum_{a=1}^{k-1} Q_{ns}(a,b|x,y)  = Q_{ns}(k,b|x,y). \]
Thus $Q_{ns}(k,b|x,y) \in \Span B$. A similar observation shows that $Q_{ns}(a,k|x,y) \in \Span B$ for each $a \in [k]$ and $x,y \in [n]$. Finally let $x,y \in [n]$. Then
\[ e_{ns} - \sum_{a,b=1}^{k-1} Q_{ns}(a,b|x,y)  -  \sum_{b=1}^{k-1} Q_{ns}(k,b|x,y)  -  \sum_{a=1}^{k-1} Q_{ns}(a,k|x,y)  = Q_{ns}(k,k|x,y). \] 
Thus $Q_{ns}(k,k|x,y) \in \Span B$. We conclude that $\dim(\mathcal{V}_{ns}) \leq |B| = (n(k-1)+1)^2$. To show that $\dim(\mathcal{V}_{ns}) \geq |B|$, let $\mathcal{V}$ denote the nonsignalling vector space from Example \ref{ex: commutative ns vector space}. By the universal property of $\mathcal{V}_{ns}$, there exists a linear surjection from $\mathcal{V}_{ns}$ to $\mathcal{V}$. Hence $\dim(\mathcal{V}_{ns}) \geq \dim(\mathcal{V})= (n(k-1)+1)^2$. Therefore $\dim(\mathcal{V}_{ns}) = (n(k-1)+1)^2 = |B|$, and thus $B$ is a basis for $\mathcal{V}_{ns}$.
\end{proof}

The observation that $\dim(\mathcal{V}_{ns}) = (n(\mathcal{V}_{ns})(k(\mathcal{V}_{ns}) -1)+1)^2$ implies that $\mathcal{V}_{ns}$ is isomorphic to the nonsignalling vector space from Example \ref{ex: commutative ns vector space}, since vector spaces are unique up to dimension.

Having established the structure of the underlying vector space $\mathcal{V}_{ns}$, we seek to construct positive cones $D_{ns}$ and $D_{qc}$ making $(\mathcal{V}_{ns}, D_{ns}, e_{ns})$ and $(\mathcal{V}_{ns}, D_{qc}, e_{ns})$ universal AOU spaces with respect to nonsignalling and quantum commuting correlations, respectively. We will establish the necessary universal properties by constructing corresponding matrix orderings $\mathcal{D}_{ns}$ and $\mathcal{D}_{qc}$ with $(\mathcal{D}_{ns})_1 = D_{ns}$ and $(\mathcal{D}_{qc})_1 = D_{qc}$ so that $(\mathcal{V}_{ns}, \mathcal{D}_{ns}, e_{ns})$ and $(\mathcal{V}_{ns}, \mathcal{D}_{qc}, e_{ns})$ are universal nonsignalling and quantum commuting operator systems, respectively. We begin with the nonsignalling case.

\begin{definition} \label{defn: ns cones}
Let $n,k \in \mathbb{N}$, and let $\mathcal{V}_{ns}$ denote the universal nonsignalling vector space with $n(\mathcal{V}_{ns})=n$ and $k(\mathcal{V}_{ns})=k$. We define
\[ (\mathcal{V}_{ns})_h := \left\{ \sum t(a,b|x,y)Q_{ns}(a,b|x,y) : t(a,b|x,y) \in \mathbb{R} \text{ for all } a,b \in [k], x,y \in [n] \right\} \]
and
\[ D_{ns} := \left\{ \sum t(a,b|x,y)Q_{ns}(a,b|x,y) : t(a,b|x,y) \geq 0 \text{ for all } a,b \in [k], x,y \in [n] \right\} \]
so that $D_{ns}$ is the smallest cone generated by the vectors $\{Q_{ns}(a,b|x,y)\}$. Given an arbitrary element $z = \sum r(a,b|x,y)Q_{ns}(a,b|x,y)$ with $r(a,b|x,y) \in \mathbb{C}$, define $z^* := \sum \overline{r(a,b|x,y)}Q_{ns}(a,b|x,y)$. Let $\mathcal{D}_{ns}$ denote the maximal matrix ordering $\mathcal D_{ns}^{max}$, so that $(\mathcal{D}_{ns}^{max})_1 = D_{ns}$.
\end{definition}

We first establish that $(\mathcal{V}_{ns}, D_{ns}, e_{ns})$ is an AOU space. It will follow that $(\mathcal{V}_{ns}, \mathcal{D}_{ns}, e_{ns})$ is an operator system.

\begin{proposition}
Let $n,k \in \mathbb{N}$, and let $\mathcal{V}_{ns}$ denote the universal nonsignalling vector space with $n(\mathcal{V}_{ns})=n$ and $k(\mathcal{V}_{ns})=k$. Then $(\mathcal{V}_{ns}, D_{ns}, e_{ns})$ is an AOU space.
\end{proposition}

\begin{proof}
First, observe that the map $*:\mathcal{V}_{ns} \to \mathcal{V}_{ns}$ from Definition~\ref{defn: ns cones} is a well-defined involution on $\mathcal{V}_{ns}$. Indeed, this involution coincides with the canonical involution on the nonsignalling vector space $\mathcal{V}$ from Example \ref{ex: commutative ns vector space}. It is clear that $(\mathcal{V}_{ns})_h = \{x \in \mathcal{V}_{ns} : x=x^*\}$ and that $D_{ns} \subseteq (\mathcal{V}_{ns})_h$. To complete the proof, we only need to establish that $e_{ns}$ is an interior point of the cone $D_{ns}$ (c.f. \cite[Lemma 1.7]{Aliprantis2007ConesAD}). Since the extreme points of $D_{ns}$ are precisely the rays $\{t Q_{ns}(a,b|x,y) : t \in [0,\infty)\}$ and since $e_{ns}$ can be expressed as a convex combination of nonzero points along these rays, for example as
\[ e_{ns} = \frac{1}{n^2 k^2} \sum_{a,b \in [k], x,y \in [n]} k^2 Q_{ns}(a,b|x,y), \]
we see that $e_{ns}$ is an interior point for $D_{ns}$. Therefore $(\mathcal{V}_{ns}, D_{ns}, e_{ns})$ is an AOU space.
\end{proof}

The next proposition shows that the operator system $(\mathcal{V}_{ns}, \mathcal{D}_{ns}, e_{ns})$ is universal with respect to all nonsignalling operator systems.

\begin{proposition} \label{prop: D_ns universal}
Let $n,k \in \mathbb{N}$, and let $\mathcal{V}_{ns}$ denote the universal nonsignalling vector space with $n(\mathcal{V}_{ns})=n$ and $k(\mathcal{V}_{ns})=k$. Suppose that $(\mathcal{W}, \mathcal{C}, e)$ is a nonsignalling operator system with $n(\mathcal{W})=n$, $k(\mathcal{W})=k$, and generators $\{Q(a,b|x,y)\}$. Then the map $\pi: (\mathcal{V}_{ns}, \mathcal{D}_{ns}, e_{ns}) \to (\mathcal{W}, \mathcal{C}, e)$ defined by $\pi(Q_{ns}(a,b|x,y)) = Q(a,b|x,y)$ is unital completely positive.
\end{proposition}

\begin{proof}
Since $\mathcal{W}$ is a nonsignalling vector space, the map $\pi$ is well-defined, linear, and unital by Proposition \ref{prop: universal non-signalling vector space}. By Definition \ref{defn: qc operator system}, each operator $Q(a,b|x,y)$ is positive, and thus \[ \sum t(a,b|x,y) Q(a,b|x,y) \geq 0 \] whenever $t(a,b|x,y) \geq 0$ for each $a,b \in [k]$ and $x,y \in [n]$. It follows that $\pi$ is a unital positive map between the AOU spaces $(\mathcal{V}_{ns}, D_{ns}, e_{ns})$ and $(\mathcal{W}, \mathcal{C}_1, e)$. Hence $\pi$ is a unital positive map on the operator system $(\mathcal{V}_{ns}, \mathcal{D}_{ns}, e_{ns})$. Since $\mathcal{D}_{ns} = D_{ns}^{max}$, $\pi$ is completely positive.
\end{proof}

The following theorem shows that the AOU space $(\mathcal{V}_{ns}, D_{ns}, e_{ns})$ is precisely the affine dual of the convex set $C_{ns}(n,k)$ when $n(\mathcal{V}_{ns})=n$ and $k(\mathcal{V}_{ns})=k$.

\begin{theorem} \label{Characterize non-signalling correlations}
Let $n,k \in \mathbb{N}$, let $\mathcal{V}_{ns}$ denote the universal nonsignalling vector space with $n(\mathcal{V}_{ns})=n$ and $k(\mathcal{V}_{ns})=k$, and let $\{Q_{ns}(a,b|x,y)\}$ be the set of generators of $\mathcal V_{ns}.$ If $p = \{p(a,b|x,y)\}$ is a correlation, then $p \in C_{ns}(n,k)$ if and only if there exists a state $\phi:(\mathcal{V}_{ns}, D_{ns}, e_{ns}) \to \mathbb{C}$ such that $p(a,b|x,y) = \phi(Q_{ns}(a,b|x,y))$ for each $a,b \in [k], x,y \in [n]$. 
\end{theorem}

\begin{proof}
First suppose $p \in C_{ns}(n,k)$. Then by Theorem~\ref{thm: NS and QC operator systems}, there exists a nonsignalling operator system $(\mathcal{W}, \mathcal{C}, e)$ with generators $\{Q(a,b|x,y)\}$ and a state $\phi: \mathcal{W} \to \mathbb{C}$ such that $p(a,b|x,y) = \phi(Q(a,b|x,y))$. By Proposition \ref{prop: D_ns universal}, the mapping $\pi: \mathcal{V}_{ns} \to \mathcal{W}$ given by $\pi(Q_{ns}(a,b|x,y)) = Q(a,b|x,y)$ is a unital positive map. Hence $\phi \circ \pi: \mathcal{V}_{ns} \to \mathbb{C}$ is a state satisfying $p(a,b|x,y) = \phi \circ \pi(Q_{ns}(a,b|x,y))$.

On the other hand, suppose $\phi: (\mathcal{V}_{ns}, D_{ns}, e_{ns}) \to \mathbb{C}$ is a state. Then $\phi$ is also a state on the operator system $(\mathcal{V}_{ns}, \mathcal{D}_{ns}, e_{ns})$. Since this operator system is nonsignalling, it follows from Theorem \ref{thm: NS and QC operator systems} that $p(a,b|x,y) := \phi(Q_{ns}(a,b|x,y))$ defines a nonsignalling correlation.
\end{proof}

We will now pursue the analogous results for quantum commuting operator systems and quantum commuting correlations.

\begin{definition} \label{defn: qc cones}
Let $n,k \in \mathbb{N}$ and let $\mathcal{V}_{ns}$ denote the universal nonsignalling vector space with $n(\mathcal{V}_{ns})=n$, and $k(\mathcal{V}_{ns})=k$. We define
\[ \mathcal{D}_{qc} := \mathcal{D}_{ns}(p_1, \dots, p_N)^\infty \]
and 
\[ D_{qc} := \mathcal{D}_{ns}(p_1, \dots, p_N)_1^\infty = (\mathcal{D}_{qc})_1, \]
where $N = n^2k^2$ and $\{p_1, p_2, \dots, p_N\}$ is some enumeration of the generators $\{Q_{ns}(a,b|x,y)\}$.
\end{definition}

A priori, the cone $D_{qc}$ and the matrix ordering $\mathcal{D}_{qc}$ depend on a choice of ordering $\{p_1, p_2, \dots, p_N\}$ for the generators $\{Q(a,b|x,y)\}$. However we will see in Corollary~\ref{cor: Embedding of v_ns into B(H)} that the $\mathcal{D}_{qc}$ and $D_{qc}$ are the same regardless of which ordering $\{p_1, p_2, \dots, p_N\}$ is chosen.

In the following proof, we let 
\[ P_i^N(x) := I_{2^{i-1}} \otimes (x \oplus x^{\perp}) \otimes J_{2^{N-i}} \quad \text{and} \quad P_{i,j}^{N,L}(x) := I_{2^{N(j-1)}} \otimes P_i^N(x) \otimes J_{2^{N(L-j)}}  \]
for any positive contraction $x$, where $x^{\perp} := e - x$, $i \in [N]$, and $j \in [L]$.

\begin{proposition} \label{prop: Universal property of D_qc}
Let $n,k \in \mathbb{N}$, and let $H$ be a Hilbert space. Suppose that $\{E_{x,a}\}_{a=1}^k, \{F_{y,b}\}_{b=1}^k \subseteq B(H)$ are projection-valued measures and $E_{x,a} F_{y,b} = F_{y,b} E_{x,a}$ for each $a,b, \in [k],$ and $x,y \in [n]$. Let $\mathcal{W} := \Span \{E_{x,a} F_{y,b} \}$. Then the map $\pi:\mathcal{V}_{ns} \to \mathcal{W}$ defined by $\pi(Q_{ns}(a,b|x,y)) = E_{x,a} F_{y,b}$ is unital completely positive.
\end{proposition}

\begin{proof}
First observe that $\mathcal{W}$ is a quantum commuting operator system with generators $E_{x,a}F_{y,b}$ and unit $I$. By Proposition \ref{prop: D_ns universal}, the mapping $\pi: (\mathcal V_{ns}, \mathcal D_{ns}, e_{ns}) \to \mathcal W,$ defined by $\pi(Q_{ns}(a,b|x,y))=E_{x,a} F_{y,b}$ is unital completely positive with respect to the matrix ordering $\mathcal{D}_{ns}$. It remains to show that $\pi$ is unital completely positive with respect to the matrix ordering $\mathcal{D}_{qc}$.

Let $N := n^2 k^2$, and for each $L > 0$ let $P_{i,j}^{N,L} := P_{i,j}^{N,L}(p_i)$ and $Q_{i,j}^{N,L} := Q_{i,j}^{N,L}(p_i)$, where $\{p_1, p_2, \dots, p_N\}$ is the enumeration of the positive contractions $\{Q_{ns}(a,b|x,y)\}$ used in Definition~\ref{defn: qc cones}. Likewise, let $\widetilde{P}_{i,j}^{N,L} := P_{i,j}^{N,L}(\widetilde{p}_i)$ and $\widetilde{Q}_{i,j}^{N,L} := Q_{i,j}^{N,L}(\widetilde{p}_i)$, where $\{\widetilde{p}_1, \dots, \widetilde{p}_N\}$ denotes the corresponding enumeration of the $n^2 k^2$ projections $\{ E_{x,a} F_{y,b} \}$ in $\mathcal{W}$. 

Let $x \in (\mathcal{D}_{qc})_n$. By the definition of $\mathcal{D}_{qc}$, we see that for every $\epsilon > 0$ there exists an integer $L > 0$ such that for any $N \times L$ matrix $\{\epsilon_{i,j} \}$ of strictly positive real numbers there exists a corresponding $N \times L$ matrix $\{t_{i,j}\}$ of strictly positive real numbers such that
\[ (x + \epsilon I_n \otimes e_{ns}) \otimes J_{2^{NL}} + \sum_{i,j} \epsilon_{i,j} I_n \otimes P_{i,j}^{N,L} + \sum_{i,j} t_{i,j} I_n \otimes Q_{i,j}^{N,L} \in (\mathcal D_{ns})_{n2^{NL}}. \]
By applying $\pi_{n2^{NL}}$ to this expression, Proposition \ref{prop: D_ns universal} implies that
\[ (\pi_n(x) + \epsilon I_n \otimes I_H) \otimes J_{2^{NL}} + \sum_{i,j} \epsilon_{i,j} I_n \otimes \widetilde{P}_{i,j}^{N,L} + \sum_{i,j} t_{i,j} I_n \otimes \widetilde{Q}_{i,j}^{N,L} \in  B(H^{n2^{NL}})^+. \]
Since each $E_{x,a} F_{y,b}$ is a projection, it follows from Proposition \ref{prop: multiple projection characterization concrete case} that $\pi_n(x) + \epsilon I_n \otimes I_H \geq 0$. It then follows from the Archimedean property that $\pi_n(x) \geq 0$. Hence $\pi$ is unital completely positive with respect to $\mathcal{D}_{qc}$.
\end{proof}

We have not yet established that $D_{qc}$ is a proper cone. Using Proposition \ref{prop: Universal property of D_qc}, we can now obtain this result as a corollary.

\begin{corollary} \label{cor: Embedding of v_ns into B(H)}
Let $n,k \in \mathbb{N}$ and let $\mathcal{V}_{ns}$ denote the universal nonsignalling vector space with $n(\mathcal{V}_{ns})=n$ and $k(\mathcal{V}_{ns})=k$. Then $D_{qc}$ is a proper cone. Hence $(\mathcal{V}_{ns}, D_{qc}, e_{ns})$ is an AOU space, and $(\mathcal{V}_{ns}, \mathcal{D}_{qc}, e_{ns})$ is a quantum commuting operator system. Moreover, the matrix ordering $\mathcal{D}_{qc}$ is independent of the choice of enumeration $\{p_1, p_2, \dots, p_N\}$ of the generators $\{Q_{ns}(a,b|x,y)\}$ used in Definition \ref{defn: qc cones}.
\end{corollary}

\begin{proof}
We first establish that $(\mathcal{V}_{ns}, \mathcal{D}_{qc}, e_{ns})$ is a quantum commuting operator system. By Corollary \ref{cor: projection envelope of an AOU space} we see that $(\mathcal{V}_{ns}, \mathcal{D}_{qc}, e_{ns})$ is a quantum commuting operator system if and only if the cone $D_{qc}$ is proper. To show that $D_{qc}$ is proper, it suffices to show that there exists a quantum commuting operator system $\mathcal{W} \subseteq B(H)$ with generators $\{Q(a,b|x,y)\}$ such that $\dim(\mathcal{W}) = (n(k-1)+1)^2$. Indeed, by Proposition \ref{prop: Universal property of D_qc} the map $\pi:\mathcal{V}_{ns} \to \mathcal{W}$ defined by $\pi(Q_{ns}(a,b|x,y)) =Q(a,b|x,y)$ is positive, and hence $\pi$ must map $D_{qc} \cap -D_{qc}$ to $\{0\}$. However $\dim(\mathcal{V}_{ns}) = \dim(\mathcal{W})$, so $\pi$ is injective, and hence $D_{qc} \cap -D_{qc} = \{0\}$. For an example of such an operator system, one can take the nonsignalling vector space from Example \ref{ex: commutative ns vector space} equipped with the canonical unit and order structure inherited from its natural embedding into $B(\mathbb{C}^{k^{2n}})$. Since $(\mathcal{V}_{ns}, \mathcal{D}_{qc}, e_{ns})$ is an operator system, it follows that $(\mathcal{V}_{ns}, D_{qc}, e_{ns})$ is an AOU space.

Finally, let $\{p_1, \dots, p_N\}$ and $\{q_1, \dots, q_n\}$ be two enumerations of the set of generators of $\mathcal{V}_{ns}$. Let $\mathcal{D}_{qc}$ and $\mathcal{D}_{qc}'$ denote the corresponding matrix orderings, as described in Definition \ref{defn: qc cones}. By Proposition \ref{prop: Universal property of D_qc}, the identity map is unital completely positive whether regarded as a map from $(\mathcal{V}_{ns}, D_{qc}, e_{ns})$ to $(\mathcal{V}_{ns}, D_{qc}', e_{ns})$ or from $(\mathcal{V}_{ns}, D_{qc}', e_{ns})$ to $(\mathcal{V}_{ns}, D_{qc}, e_{ns})$. It follows that $\mathcal{D}_{qc} = \mathcal{D}_{qc}'$.
\end{proof}

We now show $C_{qc}(n,k)$ is affinely isomorphic to the state space of $(\mathcal{V}_{ns}, D_{qc}, e_{ns})$ when $n(\mathcal{V}_{ns})=n$ and $k(\mathcal{V}_{ns})=k$.

\begin{theorem} \label{thm: qc correlation characterization}
Let $n,k \in \mathbb{N}$ and let $\mathcal{V}_{ns}$ denote the universal nonsignalling vector space with $n(\mathcal{V}_{ns})=n$ and $k(\mathcal{V}_{ns})=k$. If $p = \{p(a,b|x,y)\}$ is a correlation, then $p \in C_{qc}(n,k)$ if and only if there exists a state $\phi:(\mathcal{V}_{ns}, D_{qc}, e_{ns}) \to \mathbb{C}$ such that $p(a,b|x,y) = \phi(Q(a,b|x,y))$ for each $a,b \in [k], x,y \in [n]$.
\end{theorem}

\begin{proof}
First suppose that $p \in C_{qc}(n,k)$. Then by Theorem~\ref{thm: NS and QC operator systems}, there exists a quantum commuting operator system $(\mathcal{W}, \mathcal{C}, e)$ with generators $\{Q(a,b|x,y)\}$ and a state $\phi: \mathcal{W} \to \mathbb{C}$ such that $p(a,b|x,y) = \phi(Q(a,b|x,y))$. By Proposition~\ref{prop: Universal property of D_qc}, the mapping $\pi: \mathcal{V}_{ns} \to \mathcal{W}$ given by $\pi(Q_{ns}(a,b|x,y)) = Q(a,b|x,y)$ is a unital positive map with respect to the cone $D_{qc}$. Hence $\phi \circ \pi: \mathcal{V}_{ns} \to \mathbb{C}$ is a state on $(\mathcal{V}_{ns}, D_{qc}, e_{ns})$ satisfying $p(a,b|x,y) = \phi \circ \pi(Q_{ns}(a,b|x,y))$.

Conversely, suppose that $\phi: (\mathcal{V}_{ns}, D_{qc}, e_{ns}) \to \mathbb{C}$ is a state. Then $\phi$ is also a state on the operator system $(\mathcal{V}_{ns}, \mathcal{D}_{qc}, e_{ns})$. Since this operator system is quantum commuting, it follows from Theorem~\ref{thm: NS and QC operator systems} that $p(a,b|x,y) := \phi(Q_{ns}(a,b|x,y))$ defines a quantum commuting correlation.
\end{proof}

\begin{remark}
We conclude with some remarks on generalizations of the above ideas that immediately follow from our work. First, our notions of nonsignalling vector spaces and nonsignalling operator systems are readily generalized to the multipartite situation, where one considers correlations of the form $p(a_1,a_2,\dots,a_n|x_1, x_2, \dots, x_n)$. These correlations describe the scenario where $n$ spacially distinct parties each perform a measurement on their respective quantum system. One may then consider quantum commuting correlations arising from $n$ mutually commuting C*-algebras in a common Hilbert space. To describe these correlations using our ideas, we redefine nonsignalling operator systems to be generated by operators $\{Q(a_1, \dots, a_n| x_1, \dots, x_n)\}$ satisfying
\[ \sum_{a_1, \dots, a_n} Q(a_1, \dots, a_n| x_1, \dots, x_n) = e \] and having well-defined marginal operators \[ E_i(a_i|x_i) = \sum_{a_j, j \neq i} Q(a_1, \dots, a_n|x_1, \dots, x_n). \]
Requiring each generator to be an abstract projection yields a quantum commuting operator system. The constructions of universal nonsignalling and quantum commuting operator systems proceeds in the same manner as the bipartite case. Our work also readily generalizes to the setting of matricial correlation sets, as described in \cite{ozawa2013connes}. Indeed it is straightforward to see that the matrix affine dual of the matricial nonsignalling and quantum commuting correlations are precisely the operator systems $(\mathcal V_{ns}, \mathcal{D}_{ns}, e_{ns})$ and $(\mathcal V_{ns}, \mathcal{D}_{qc}, e_{ns})$, respectively, using Webster-Winkler duality \cite{webster1999krein}.
\end{remark} 

\bibliographystyle{plain}
\bibliography{References}

\end{document}